\documentclass[journal]{IEEEtran}
\usepackage{amsmath,amssymb}
\usepackage{graphicx,graphics,color,psfrag}
\usepackage{cite,balance}
\usepackage{caption}
\allowdisplaybreaks
\usepackage{algorithm}
\usepackage{algorithmic}
\usepackage{float}
\usepackage[caption=false,font=normalsize,labelfont=sf,textfont=sf]{subfig}
\usepackage{mdframed} 
\usepackage{framed} 
\usepackage{accents}
\usepackage{amsthm}
\usepackage{bm}
\usepackage{url}
\usepackage[english]{babel}
\usepackage{multirow}
\usepackage{enumerate}
\usepackage{cases}
\usepackage{stfloats}
\usepackage{dsfont}
\usepackage{color,soul}
\usepackage{amsfonts}
\usepackage{cite,graphicx,amsmath,amssymb}
\usepackage{fancyhdr}
\usepackage{hhline}
\usepackage{graphicx,graphics}
\usepackage{array,color}
\usepackage{mathtools}
\usepackage{amsmath}
\usepackage[T1]{fontenc}
\usepackage{float} 
\usepackage{diagbox}
\usepackage{tcolorbox}
\usepackage{float}
\usepackage{siunitx}
\usepackage{lipsum}
\usepackage{adjustbox}
\usepackage{stfloats}
\usepackage{bm}

\newtheorem{definition}{\emph{\underline{Definition}}}

\newtheorem{lemma}{\emph{\underline{Lemma}}}

\newtheorem{corollary}{\emph{\underline{Corollary}}}

\newtheorem{proposition}{\emph{\underline{Proposition}}}

\newtheorem{remark}{\bf \emph{\underline{Remark}}}

\def\({\left(}
\def\){\right)}

\def\b0{{\mathbf{0}}}

\renewcommand{\mod}{\tx{mod}}

\newcommand{\tx}[1]{\texttt{#1}}

\graphicspath{{./Figures/}}
\begin{document}
\captionsetup[figure]{name={Fig.}}
\title{Frequency-switching Array Enhanced Physical Layer Security: A Movable Antenna Perspective} 
\author{Cong Zhou, Changsheng You,~\IEEEmembership{Member,~IEEE}, Chao Zhou, Weidong Mei,~\IEEEmembership{Member,~IEEE}, \\Zhi Chen,~\IEEEmembership{Senior Member,~IEEE}, Chengwen Xing,~\IEEEmembership{Member,~IEEE} and Rui Zhang,~\IEEEmembership{Fellow,~IEEE}
\thanks{Cong Zhou is with the Department of Electronic and Electrical Engineering, Southern University of Science and Technology, Shenzhen 518055, China and also with the School of Electronics and Information Engineering, Harbin Institute of Technology, Harbin, 150001, China (e-mail: zhoucong@stu.hit.edu.cn).}
\thanks{Changsheng You and Chao Zhou are with the Department of Electronic and Electrical Engineering, Southern University of Science and Technology, Shenzhen 518055, China (e-mail: youcs@sustech.edu.cn; zhouchao2024@mail.sustech.edu.cn).}
\thanks{Weidong Mei and Zhi Chen are with the National Key Laboratory of Wireless Communications, University of Electronic Science and Technology of China, Chengdu 611731, China. (e-mail: wmei@uestc.edu.cn; chenzhi@uestc.edu.cn).}
\thanks{Chengwen Xing is with the School of Information and Electronics, Beijing Institute of Technology, Beijing 100081, China (e-mail: xingchengwen@gmail.com).}
\thanks{Rui Zhang is with the Department of Electrical and Computer Engineering, National University of Singapore, Singapore 117583 (e-mail: elezhang@nus.edu.sg).}
\thanks{Corresponding author: Changsheng You.}
}
\maketitle
\begin{abstract}
	In this paper, we propose a new \emph{frequency-switching array} (FSA) to enhance the physical-layer security (PLS) in the presence of multiple eavesdroppers (Eves), where the carrier frequency can be flexibly switched and small frequency offsets can be imposed on each antenna at the secrecy transmitter (Alice).
	First, we analytically show that by flexibly controlling the carrier frequency parameters, FSAs can effectively form uniform/non-uniform sparse arrays, hence resembling existing mechanically controlled \emph{movable antennas} (MAs) via the control of inter-antenna spacing and providing additional degree-of-freedom (DoF) in the beam manipulation.
	Although the proposed FSA suffers from additional path-gain attenuation in the received signals, it can overcome several hardware and signal processing issues incurred by MAs, such as limited positioning accuracy, extra hardware and energy cost.
	Then, a secrecy-rate maximization problem is formulated under the constraints on the frequency control.
	To shed useful insights, we first consider a \emph{secrecy-guaranteed} problem with a null-steering constraint for which maximum ratio transmission (MRT) beamformer is considered at Alice and the frequency offsets are set as uniform frequency increment.
	Interestingly, it is shown that the proposed FSA can flexibly realize null-steering over Eve in both the angular domain (by tuning carrier frequency) and range domain (by controlling per-antenna frequency offset), thereby achieving improved PLS performance.
	Then, for the general case, we propose an efficient algorithm to solve the formulated non-convex optimization problem by using the block coordinate descent (BCD) and projected gradient ascent (PGA) techniques.
	Finally, numerical results demonstrate that the proposed FSA achieves superior secrecy rate performance over conventional fixed-position array, while it only suffers a slight secrecy rate loss than the existing mechanically controlled MA, yet at much lower hardware and signal processing costs.
\end{abstract}

%-----------------------------------------------------第一部分
\section{Introduction}
The future sixth-generation (6G) and beyond wireless communication systems are envisioned to be capable of accommodating emerging services, such as virtual reality (VR), digital twin (DT) and artificial intelligence (AI) \cite{you2024next}.
%Specifically, thanks to ultra-large bandwidth, a data transmission rate up to a terabit per second can be achieved even without using massive multiple-input multiple-output (MIMO) technique in the fifth-generation (5G) communication systems \cite{han2022THz}.
However, owing to the broadcasting nature of wireless channels, preventing confidential information from being intercepted by unauthorized parties becomes a critical challenge.
As such, physical layer security (PLS) has emerged as a promising approach to prevent eavesdropping while ensuring reliable transmissions \cite{ma2018security}.
Although the high frequency bands enable the deployment of a large number of antennas at the transmitter to form pencil-like directional beams and thereby enhance PLS performance, wireless communication systems still face severe security threats when eavesdroppers (Eves) are located in close proximity to legitimate users, especially when they share similar spatial angles with the intended receiver \cite{ma2018security}.

To overcome this issue, movable antennas (MAs) have been recently proposed to reconfigure wireless channels between the transmitter and receiver, for which positions of antennas can be flexibly adjusted within a confined region to achieve improved rate performance \cite{shao20246d,10848372,shao2025tutorial}.
In particular, for PLS, MAs/FAs are capable of achieving null-steering over Eves while maintaining full array gains at the legitimate user \cite{zhu2023movableNullSteering}, thereby significantly enhancing PLS performance.
However, MAs generally face several challenges in practice especially in high frequency bands.
First, the mechanical control of antenna movement in MAs practically introduces high hardware and energy costs as well as bulky size, especially in high frequency bands \cite{ding2025flexible}.
Second, the rapid variation of the channel imposes strict requirements on the response latency of MAs, which further degrades the performance of MAs \cite{zhu2023movable}.

In this paper, we propose a new frequency-switching array (FSA) to enhance the performance of PLS systems, which can effectively overcome the hardware complexity of mechanically controlled MAs while maintaining comparable spatial degree-of-freedoms (DoFs) in the beam control for enhancing PLS performance.

\vspace{-6pt}
\subsection{Related Works}
PLS has been extensively investigated for fixed-position arrays (FPAs) in far-field communications.
For example, the authors in \cite{li2017secure} proposed to use zero-forcing (ZF) precoding to achieve null-steering over Eves, which effectively prevents confidential information leakage.
However, when the legitimate user locates in the vicinity of Eves (i.e., with strong channel correlation), ZF precoding may reduce the signal power received at the legitimate user.
In addition, artificial noise (AN) was utilized to enhance PLS by acting as an interference signal to Eves, while having little impact on the received signals at legitimate users \cite{yu2020robust}.
Moreover, to flexibly control spatial beams, the authors in \cite{hu2024secure} proposed to use mechanically controlled MAs to enhance PLS, where they optimized the positions of antennas and the digital beamformer to maximize the achievable secrecy rate via the projected gradient ascent method \cite{ranjan2023gradient}.
Then, this work was further extended in \cite{10623758} to scenarios where the secrecy transmitter (Alice) has no prior knowledge of the channel state information (CSI) of Eves.
Specifically, the secrecy outage probability was minimized by jointly optimizing the antennas' positions at Alice and the beamforming vector.
However, when the legitimate receiver (Bob) and the Eves are located at the same angle, MAs become ineffective.
This is because the far-field channel correlation between Bob and Eves is strongly correlated in the angular domain.
This limitation can be alleviated only if the movable region of MAs is sufficiently large, such that both Bob and Eves fall within Alice’s near-field region, and the MA can adjust its position to form distinct spatial channels for them. Nevertheless, achieving such a large movable region may be difficult in practice, especially in space-constrained  applications.

To address this issue, frequency diverse arrays (FDAs) have  emerged as a promising technology to form \emph{range-dependent} beams with small frequency offset imposed on each antenna, which can enhance PLS performance even when Bob ad Eves are located at the same angle \cite{lin2017physical}.
For example, in \cite{cheng2024secure}, the authors proposed to minimize the transmit power while maximizing the secrecy rate in an FDA-assisted secure transmission system, where different frequency offsets are imposed on different antennas.
However, when multiple Eves are considered, a larger frequency offset at Alice is required to provide more DoFs in spatial beamforming, thereby mitigating power leakage to all Eves \cite{qiu2020multi}.
However, this may not be always feasible in practice, as the frequency increment is typically assumed to be much smaller than the carrier frequency.
To tackle this issue, a new movable FDA (M-FDA) was proposed in \cite{cheng2025movable}, which combines the advantages of MAs and FDAs such that antennas' positions and frequency offsets across different antennas can be flexibly adjusted.
In particular, by exploiting extra spatial DoFs provided by MAs along with FDAs, the authors jointly optimized the beamforming vector, antennas' positions and frequency offsets to maximize the achievable secrecy rate.
However,  M-FDAs still suffer from the drawbacks of MAs in high frequency bands as mentioned above (e.g., high hardware and energy costs, as well as considerable response latency).
How to overcome the hardware limitations of MAs while maintaining their spatial DoF remains an open problem.

\vspace{-10pt}
\subsection{Motivation and Contributions}
Motivated by the above, in this paper, we propose a new flexibly FSA to enhance PLS performance, where the carrier frequency can be flexibly switched and small frequency offsets can be imposed across different antennas.
It is shown that such frequency switching is an efficient way to form electronically controlled arrays resembling MAs as both of them can effectively control the inter-antenna spacing relative to electrical length, thereby improving secrecy-rate performance.
In particular, the proposed FSA achieves more precise control of the effective antenna positions with very low or negligible antenna configuration latency, making it a promising solution for applications in high frequency bands.
The main contributions of this paper are summarized as follows.
\begin{itemize}
	\item First, we propose a new FSA to enhance PLS performance where the carrier frequency and frequency offsets imposed at different antennas are flexibly adjusted.
	Interestingly, we show that switching the carrier frequency of the proposed FSA is equivalent to forming a \emph{sparse uniform linear array} (S-ULA) with expanded array aperture, while adjusting frequency offsets across antennas corresponds to forming a \emph{sparse non-uniform linear array}.
	Then, to shed useful insights, we consider a typical one-Eve-one-Bob system and formulate a secrecy-guaranteed problem under the constraint of null-steering towards Eve.
	It is revealed that in the optimal solution, sufficiently large maximum frequency offsets applied to the antennas can flexibly achieve null-steering towards Eve, even when Bob and Eve are located at the same angle.
	However, when the frequency offsets are not sufficiently large and/or when Bob and Eve are located at the same range, null-steering can be achieved by switching the carrier frequency only.
	\item Second, for the general system setup with multiple Eves, we formulate an optimization problem to maximize the secrecy rate at Bob under the practical frequency control constraints.
	To solve this non-convex problem, we leverage the block coordinate descent (BCD) technique to iteratively optimize the transmit beamformer, frequency increment vector (FIV) and carrier frequency, respectively.
	Then, the optimal solution to the subproblem of transmit beamformer optimization is obtained by the generalized Rayleigh quotient, while the other two subproblems are solved by the projected gradient ascent method. 
	\item Finally, numerical results are presented to demonstrate the effectiveness of the proposed FSA in enhancing PLS performance.
	It is shown that the proposed FSA achieves a higher secrecy rate than that of FPAs and FDAs, and is only slightly outperformed by MAs.
	This is because, although the proposed FSA achieves flexible beam control comparable to that of MAs, it experiences larger power attenuation due to the increased carrier frequency. Nevertheless, compared to FDAs and FPAs, the proposed FSA offers much higher beam-control flexibility, thereby achieving better secrecy-rate performance.
\end{itemize}
\vspace{-10pt}
\section{System Model}
We consider an FSA-enabled PLS communication system, where Alice transmits confidential information to a Bob in the presence of $M$ Eves\footnote{The proposed FSA can also enhance PLS performance of multiple Bobs by maximizing the sum-secrecy-rate of all Bobs \cite{liu2025physical} and the proposed projected gradient ascent (PGA) method in Section \ref{Sec:solving optimization problem} can be directly applied to solve this problem.}.
A uniform linear array (ULA) comprising $N$ antennas is equipped at Alice, while Bob and $M$ Eves are equipped with a single antenna.
In addition, both Bob and $M$ Eves are assumed to be located in the far-field region of Alice.

\vspace{-9pt}
\subsection{Frequency-Switching Array}
For the considered FSA, let $d_0 = \frac{\lambda_0}{2}$ denote the inter-antenna spacing where $\lambda_0$ represents the carrier wavelength.
In particular, the basic carrier frequency is denoted by $f_0 = \frac{c}{\lambda_0}$, where $c = 3\times 10^8$ m/s represents the light speed.
Without loss of generality, we assume that the number of antennas at Alice is an odd integer and the ULA is centered at the origin.
As such, the coordinate of the $n$-th antenna at the ULA is $(0, y_n = \delta_nd_0)$, where $\delta_n = \frac{2n-N-1}{2}, n \in \mathcal{N}\triangleq\{1,2,\cdots,N\} $.

Unlike conventional MA systems where each antenna's position in the ULA can be flexibly adjusted, in the proposed FSA system, the position of all antennas are fixed, while the carrier frequency and the radiation frequency of each antenna can be flexibly switched in a given range to enhance PLS. 
It will be shown in Section \ref{sec:physical meaning} that such frequency switching is equivalent to \emph{repositioning} the antennas in the fixed ULA, resembling the MA.

%\begin{figure}[t]
%	\centering
%	\includegraphics[width=0.5\textwidth]{./Figures/SystemModels.pdf}
%	\caption{Frequency-switching array enhanced PLS communication systems.} 
%	\label{Fig:System_Model}
%	\vspace{-10pt}
%\end{figure}
Specifically, the carrier frequency can be switched, ranging from the basic frequency $f_0$ to the maximum frequency $f_{\rm H}$, while different antennas can apply additional small frequency increments/offsets independently ~\cite{cheng2024secure}.
Mathematically, the frequency of the $n$-th antenna is given by $f_n = f_{\rm c} + \Delta_{f_n}, n\in\mathcal{N}$, where $f_{\rm c}$ denotes the tunable carrier frequency common to all antennas, satisfying $f_0 \le f_{\rm c} \le f_{\rm H}$.
In addition, $\Delta_{f_n}$ is the frequency increment at the $n$-th antenna and we assume that $0 \le \Delta_{f_n} \le \Delta_{f_{\max}}$ with $\Delta_{f_{\max}} \ll f_{\rm H} $.

\vspace{-9pt}
\subsection{Channel Model}
We assume that Bob and each Eve $m$ are located at $ (\theta_{0}, r_{0}) $ and $ (\theta_{m}, r_{m}), m \in \mathcal{M} \triangleq \{1,2,\cdots,M\}$, where $\theta_{0}$ ($\theta_{m}$) and $r_{0}$ ($r_{m}$) denote the spatial angle and range for the Bob ($m$-th  Eve). Moreover, Alice is assumed to have the perfect CSI of Eves, while the CSI acquisition methods for Eves and robust beamforming design under imperfect CSI will be discussed in Section \ref{sec:Extension to CSI Acquisition for Eve}.
%Since both Bob and Eves are located in the far-field region of Alice, we adopt the planar wavefront model to characterize channel models of Bob and Eves.
Let ${s}_{\rm B}(t)$ denote the confidential baseband signal for the desired Bob. Then, the passband signal transmitted by the $n$-th antenna at Alice and received by Bob is given by
\begin{equation}
	y_n(t) = {\rm Re}\big\{\gamma_{\rm c} g_{{\rm B},f_0}e^{\jmath 2\pi (f_{\rm c}+\Delta_{f_n})(t - \frac{r_n}{c})}{s}_{\rm B}(t)\big\}, 
\end{equation}
where $g_{{\rm B},f_0}$ denotes the path gain of the LoS path with respect to (w.r.t.) the basic frequency $f_0$, while $\gamma_{\rm c}=\frac{f_0}{f_{\rm c}}$ represents the \emph{attenuation factor} induced by additional path loss of the switching carrier frequency.
In addition, $r_n \approx r_0 - n\delta_n d_0\sin\theta_0$ denotes the range between the $n$-th antenna at Alice and Bob.
As such, by removing the carrier signal $e^{\jmath 2\pi f_{\rm c}t}$, the equivalent baseband signal can be obtained as
\begin{align}
	y_n^{\prime} &= \gamma_{\rm c} g_{{\rm B},f_0}e^{\jmath 2\pi\big((f_{\rm c}+\Delta_{f_n})\frac{r_n}{c} - {\Delta_{f_n}}t\big)}{s}_{\rm B}(t),   \label{eq:time varying signal}\\
	&=\!\gamma_{\rm c} g_{{\rm B},f_0}e^{\jmath 2\pi\big(\frac{(f_{\rm c} \!+\! \Delta_{f_n})\delta_n d_0\sin\theta_0}{c} - \frac{f_{\rm c} + \Delta_{f_n}}{c}r_0 + \Delta_{f_n}t\big)} {s}_{\rm B}(t).\notag
\end{align}

It is observed that the phase of $y_n^{\prime}$ varies over time, which is generally difficult for beamforming in practice.
Nevertheless, this problem can be addressed by using the receive processing chain in \cite{xu2021range} (to be discussed in Section \ref{sec:Transmission Protocol Design}), which ensures that the reference phase within a coherence time block remains time-invariant.
Hence, for brevity, the time-varying components ($2\pi \Delta_{f_n}t$) in the channel and signal modelling are ignored.
In this paper, we mainly consider PLS in high frequency bands, for which the NLoS channel paths exhibit negligible power due to the severe path-loss and shadowing.
As such, the LoS channel modelling is adopted for Bob and $M$ Eves.
Let $\mathbf{h}^H_{\rm B}(f_{\rm c}, \boldsymbol{\Delta}_f) \in \mathbb{C}^{1 \times N}$ denote the channel from Alice to Bob, which is modeled as 
\begin{align}
	\mathbf{h}^H_{\rm B}(f_{\rm c}, \boldsymbol{\Delta}_f) &= \sqrt{N} \gamma_{\rm c} g_{{\rm B},f_0} e^{\jmath\frac{2\pi f_{\rm c}r_0}{c}}\mathbf{a}^H(f_{\rm c}, \boldsymbol{\Delta}_f, \theta_{0}, r_0).
\label{eq:channel steering vector}
\end{align}
Herein, the channel steering vector $\mathbf{a}(f_c, \boldsymbol{\Delta}_f, {\theta}, r)$ is given by  \cite{antonik2006frequency, lan2016range}
\begin{equation}
	\begin{aligned}
		\big[\mathbf{a}(f_{\rm c}, \boldsymbol{\Delta}_f,\theta_0,r_0)\big]_n \!=\! \frac{1}{\sqrt{N}} e^{\jmath{2\pi}\big(\frac{f_{\rm c}\delta_n d_0\sin\theta_0}{c} \!+\! \frac{\Delta_{f_n}}{c}(\delta_n d_0\sin\theta_0 - r_0)\big)},
	\end{aligned}
	\notag
\end{equation}
where $\boldsymbol{\Delta}_f = [\Delta_{f_1}, \Delta_{f_2},\cdots, \Delta_{f_n}]^T$ denotes the FIV.
Since $\Delta_{f_n} \ll f_{\rm c}$, the phase term $e^{\jmath2\pi\frac{\Delta_{f_n}}{c}\delta_n d_0\sin\theta_0}$ can be neglected \cite{wang2015frequency}.
For example, given $N = 16$, $f_{\rm c} = 60$ GHz and $\Delta_{f_n} = 1$ MHz, we have $\max\limits_{\theta_0} e^{\jmath2\pi\frac{\Delta_{f_n}}{c}\delta_n d_0\sin\theta_0} = 8.4\times 10^{-4} \ll 2\pi $.
Hence, the channel steering vector in \eqref{eq:channel steering vector} can be approximated as~\cite{lan2016range}
\begin{equation}
	\label{eq:steering vector}
	\begin{aligned}
		\big[\mathbf{a}(f_{\rm c}, \boldsymbol{\Delta}_f,\theta_0,r_0)\big]_n \!=\! \frac{1}{\sqrt{N}} e^{\jmath{2\pi}\big(\frac{f_{\rm c}\delta_n d_0\sin\theta_0}{c}-\frac{\Delta_{f_n}}{c}r_0\big)},\forall n\in \mathcal{M}.
	\end{aligned}
\end{equation}

Similarly, the channel from Alice to Eve $m$, denoted by $\mathbf{h}^H_{{\rm E},m}(f_{\rm c}, \boldsymbol{\Delta}_f, t) \in \mathbb{C}^{1 \times N}$, can be modeled as
\begin{align}
	\mathbf{h}^H_{{\rm E},m}(f_{\rm c}, \boldsymbol{\Delta}_f) \!=\! \sqrt{N} \gamma_{\rm c}g_{{\rm E},f_0,m} e^{\jmath\frac{2\pi f_{\rm c}r_m}{c}}\mathbf{a}^H(f_{\rm c}, \boldsymbol{\Delta}_f,\! \theta_{m},\! r_m),
	\notag
\end{align}
where $ g_{{\rm E},f_0,m} $ represents the complex-valued path gain in the basic frequency and $ \mathbf{a}^H(f_{\rm c}, \boldsymbol{\Delta}_f, \theta_{m}, r_{m}) $ denotes the far-field channel steering vector, which can be defined similar to~\eqref{eq:steering vector}.

\vspace{-9pt}
\subsection{Signal Model}
We consider that at Alice, the digital beamforming architecture is adopted for the proposed FSA.
Let $\mathbf{w}_{\rm B} \in \mathbb{C}^{N \times 1}$ denote the digital beamformer at Alice for transmitting confidential information to Bob.
As such, the received signal at Bob is given by $y_{{\rm B}} = \mathbf{h}^H_{\rm B}(f_{\rm c}, \boldsymbol{\Delta}_f)\mathbf{w}_{\rm B}s_{\rm B}(t) + z_{{\rm B}}$, where $ z_{{\rm B}} \sim \mathcal{CN}(0,\sigma_{\rm B}^2)$ denotes the additive white Gaussian noise (AWGN) with $\sigma_{\rm B}^2$ representing the noise power at Bob.
Then, the achievable rate at Bob, in bits per second per Hertz (bps/Hz), is given by 
\begin{equation}
	R_{\rm B}
	= \log_2\Big( 1 + \frac{ |\mathbf{h}^H_{\rm B}(f_{\rm c}, \boldsymbol{\Delta}_f)\mathbf{w}_{\rm B}|^2}{\sigma_{\rm B}^2}\Big).
\end{equation}

Meanwhile, the received signal at Eve $m$ for intercepting the information of ${s}_{\rm B}(t)$ is given by $y_{{\rm E},m} = \mathbf{h}^H_{{\rm E},m}(f_{\rm c}, \boldsymbol{\Delta}_f)\mathbf{w}_{\rm B}s_{\rm B}(t) + z_{{{\rm E},m}}, \forall m\in \mathcal{M}$, where $ z_{{{\rm E},m}} \sim \mathcal{CN}(0,\sigma_{{\rm E}}^2)$ is the AWGN at Eve $m$.
Without loss of generality, we assume that the noise powers at Bob and the eavesdroppers are identical, i.e., $\sigma_{{\rm E}} = \sigma_{{\rm B}} \triangleq \sigma_{0}$.
We consider a challenging case where $M$ Eves cooperatively intercept the confidential signal transmitted to Bob.
As such, the eavesdropping rate of $M$ Eves for wiretapping ${s}_{\rm B}$ is given by \cite{hu2024secure}
\begin{equation}
	R_{{\rm E}}
	= \log_2\bigg( 1 + \frac{\sum_{m=1}^{M}|\mathbf{h}^H_{{\rm E},m}(f_{\rm c}, \boldsymbol{\Delta}_f)\mathbf{w}_{\rm B}|^2}{\sigma_0^2}\bigg).
\end{equation}
Then, the achievable secrecy rate (in bps/Hz) is given by \cite{zheng2022physical}
\begin{equation}
	\begin{aligned}
		&R_{{\rm FSA}}(\mathbf{w}_{\rm B}(t), \boldsymbol{\Delta}_f, f_{\rm c}) = [R_{\rm B} - R_{{\rm E}}]^+\\
		&~=\! \bigg[\log_2\!\bigg(\frac{|\mathbf{h}^H_{\rm B}(f_{\rm c}, \boldsymbol{\Delta}_f)\mathbf{w}_{\rm B}|^2/\sigma_0^2 \!+\! 1}{\sum_{m=1}^{M}|\mathbf{h}^H_{{\rm E},m}(f_{\rm c}, \boldsymbol{\Delta}_f)\mathbf{w}_{\rm B}|^2/\sigma_0^2 \!+\! 1}\bigg)\bigg]^+\!\!\!,
	\end{aligned}
	\label{eq:achievable secrecy rate}
\end{equation}
where $ [x]^{+} \triangleq \max \{x, 0\}$.

\begin{figure}[t]
	\vspace{-5pt}
	\centering
	\includegraphics[width=0.5\textwidth]{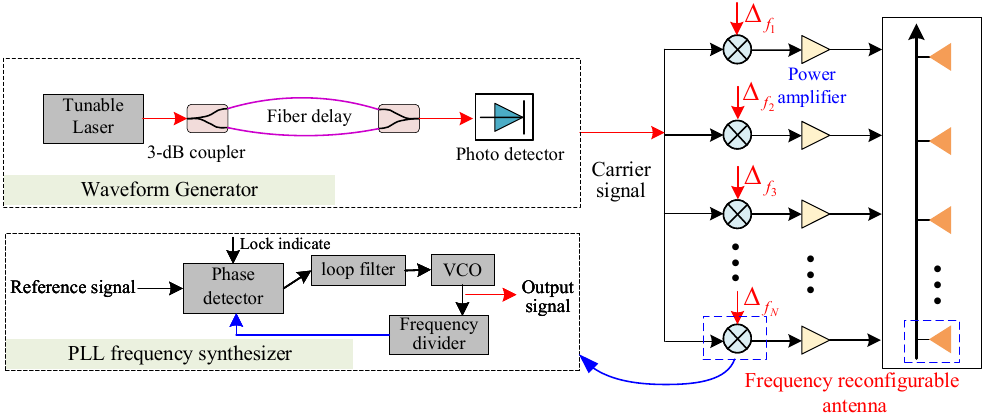}
	\caption{Hardware architecture of the proposed FSA.} 
	\label{Fig:hardware}
	\vspace{-15pt}
\end{figure}
\section{Hardware Architecture of Proposed FSA}
\label{sec:physical meaning}
In this section, we first present the hardware architecture for the proposed FSA.
Then, we analytically show that frequency switching across antennas can effectively form uniform or non-uniform sparse arrays, similar to antenna position movement.
\vspace{-20pt}
\subsection{Hardware Architecture}
The proposed FSA is composed of three main components, namely, a continuous waveform generator to adjust the carrier frequency, a phase-locked loop (PLL)-based frequency synthesizer to induce small frequency offsets, and a digital beamforming module including antennas and tunable power amplifiers (PAs).
In the following, we introduce the main components and the corresponding practical prototypes.
\begin{itemize}
	\item{\bf Continuous-wave generator:} The tunable carrier frequency function of the proposed FSA is implemented using the wavelength tunable feature of a distributed Bragg reflector (DBR) laser.
	Specifically, the output lightwave from the DBR laser is split into two optical paths by a 3-dB fiber coupler. One path propagates directly, while the other is delayed by a fiber delay line, after which the two beams are recombined, which is subsequently converted into a radio frequency (RF) signal by the high-speed photo detector (PD).
	In~\cite{zhu2013fast}, the implementation and experimental validation of such a waveform generator were reported, showing a continuously tunable carrier frequency range up to 38.45 GHz.
	\item{\bf Frequency synthesizer:} For the small frequency offset imposed on each antenna, it can be achieved by a PLL-based frequency synthesizer composed of a phase detector, a loop filter, a voltage-controlled oscillator (VCO), and a frequency divider \cite{huang2009frequency}.
	Specifically, a controllable oscillation is generated by converting the filtered phase difference signal from the phase detector into a voltage-controlled output.
	The frequency divider then scales down this output and feeds it back for comparison with the stable reference, which drives the loop into lock and ensures that the final output frequency is maintained as an integer multiple of the reference, thereby enabling precise and flexible frequency synthesis.
	A practical prototype of PLL-based frequency synthesizer used for realizing the FDA can be found in~\cite{huang2009frequency} with high frequency resolution less than 1 Hz \cite{e75-b_8_739}.
	\item{\bf Broad band devices:} Apart from modules for adjusting carrier frequency and imposing small frequency offsets, broadband devices are also crucial to implement the proposed FSA.
	For example, the authors in \cite{hussain2020compact} implemented a frequency reconfigurable antenna, which can operate at a broad band ranging from 2.05 to 10.7 GHz.
	In addition, for broad band commercial PAs, Analog Devices' ADL8124 and ADH-ALH445S can operate in the frequency bands of 1--20 GHz and 18--43 GHz, respectively.
	To achieve a PA with a wider operational frequency range, multiple narrow-band PAs can be employed to cover different frequency bands and then combined using RF switches.
\end{itemize}
\begin{figure}[t]
	\centering
	\includegraphics[width=0.4\textwidth]{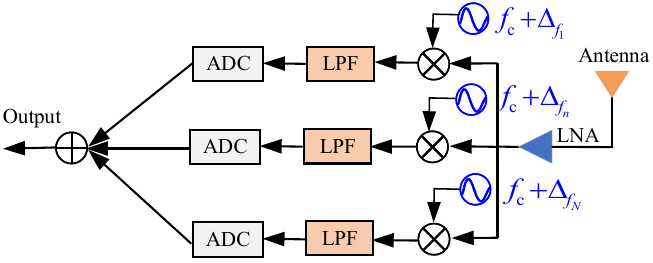}
	\caption{Illustration of the receive processing chain to remove time-varying terms.} 
	\label{Fig:Receive processing chain}
	\vspace{-15pt}
\end{figure}
\vspace{-10pt}
\subsection{Implementation Issues}
\label{sec:Transmission Protocol Design}
Introducing small frequency offsets across different antennas generally leads to a time-varying channel (see (2)), which can be compensated at the receiver side through receive processing chains.
Specifically, as illustrated in Fig. \ref{Fig:Receive processing chain}, the received signal is first amplified by a low-noise amplifier (LNA) and then divided into $N$ streams, for which the $n$-th stream is mixed with the corresponding carrier of the $n$-th transmitted antenna (i.e., $e^{\jmath 2\pi(f_{\rm c}+\Delta_{f_n})t}$).
The input signal of the $n$-th low-pass filter (LPF) after mixing is given by
\begin{equation}
	y_{\rm In}^{(n)} = \underbrace{w_ne^{\jmath \phi_n}s_{\rm B}(t)}_{\rm desired~signal} + \underbrace{\sum_{i = 1,i\neq n}^{N}w_ie^{\jmath \phi_i}e^{\jmath 2\pi(\Delta_{f_i}-\Delta_{f_n})t}s_{\rm B}(t)}_{\rm multi-branch~interference},
	\notag
\end{equation}
where $w_n$ and $\phi_n$ denote the signal combining coefficient and propagation delay phase of the $n$-th stream, respectively.
Then, the resulting $N$ signals are subsequently passed through LPFs to remove the multi-branch interference~\cite{xu2021range}. 
On one hand, considering that the baseband signal $s_{\rm B}(t)$ has a bandwidth of $B$, we should ensure $|\Delta_{f_i}-\Delta_{f_n}| \ge B, \forall i,n \in \mathcal{N}$, so that the resulting multi-branch interference can be effectively filtered out.
On the other hand, when $|\Delta_{f_i}-\Delta_{f_n}| = 0$, the output signal of the $n$-th LPF is $(w_ne^{\jmath \phi_n} + w_ie^{\jmath \phi_i})s_{\rm B}(t)$, which still yields the desired baseband signal.
As such, the above conditions can be combined as $|\Delta_{f_i}-\Delta_{f_j}|\big(|\Delta_{f_i}-\Delta_{f_j}|-B\big) \ge 0, \forall i,j \in \mathcal{N}$ to perfectly remove the time-varying components $e^{\jmath 2\pi \Delta_{f_n}t}, \forall n$ in \eqref{eq:time varying signal}.

Since conventional FDA radar is a special case of the proposed FSA, the parameter-based approach is adopted for channel estimation, where subspace-based localization methods (see, e.g.,  \cite{cui2018localization}) can be used to estimate the positions (including angle and range) of Bob and Eves for subsequent channel reconstruction.
After obtaining the CSI of Bob and Eves, secure beamforming design can be performed (as detailed in Section \ref{Sec:solving optimization problem}) to optimize the operating carrier frequency and frequency offsets imposed on individual antennas.
It is worth noting that Alice and Bob can establish an initial communication link at the basic carrier frequency $f_0$ without frequency offsets, after which Bob is informed of the operating frequency band and the antenna frequency offsets through subsequent pilot frames.
This procedure is similar to the conventional pseudorandom sequence used in frequency hopping (FH) to indicate the communication frequency to the receiver \cite{Chen2019High}.
If the FH sequence (i.e., communication frequency) is exposed to Eves, FH cannot  guarantee PLS performance.
However, the proposed FSA remains unaffected by such frequency exposure, since it ensures secrecy through spatial beam nulling.

\vspace{-8pt}
\subsection{Physical Interpretation: A Movable Antenna Perspective}
\label{sec:Physical Interpretation: A Movable Antenna Perspective}
In the following, we show that the proposed FSA resembles the conventional mechanically controlled MA via changing antenna positions.
Specifically, the channel response vector in \eqref{eq:steering vector} can be rewritten as
{\small\begin{equation}
	\notag
	\begin{aligned}
	\big[\mathbf{a}(f_{\rm c},\! \boldsymbol{\Delta}_f,\!\theta,\! r)\big]_{\!n} \!\!=\!\! \frac{1}{\sqrt{N}} \exp\!\Big({\jmath\frac{2\pi}{c}f_0{\underbrace{\big(\frac{f_{\rm c}}{f_0}\delta_n d_0\!\!-\!\!\frac{\Delta_{f_n}r}{f_0\sin\theta}\big)}_{t_n}}\sin\theta}\Big)\!,
	\end{aligned}
	\end{equation}}where $t_n = \frac{f_{\rm c}}{f_0}x_n-\frac{\Delta_{f_n}r}{f_0\sin\theta}$ denotes the equivalent position of the $n$-th antenna in the FPA with $x_n = \delta_nd_0$ representing the physical position of the $n$-th antenna.
It is observed that $t_n$ is determined by two factors: the carrier frequency switching and the frequency offsets imposed on individual antennas. Their effects are summarized as follows.

\begin{figure}[t]
	\centering
	\includegraphics[width=0.365\textwidth]{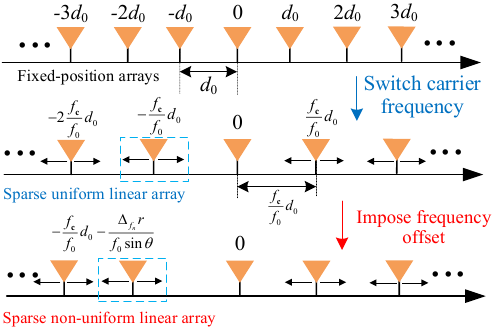}
	\caption{Illustration of the equivalence between the proposed frequency-switching array and a movable antenna array.} 
	\label{Fig:EquivalentMoveAntenna}
	\vspace{-14pt}
\end{figure}
\begin{itemize}
	\item The first term $\frac{f_{\rm c}}{f_0} x_n$ in $t_n$ arises from carrier frequency switching.
	Specifically, the position of the $n$-th antenna shifts from $x_n$ to $\frac{f_{\rm c}}{f_0} x_n$, which is scaled by a factor of $\frac{f_{\rm c}}{f_0}$.
	As a result, when there is no frequency offset across antennas, the equivalent inter-antenna spacing in the \emph{dense} FPA becomes $\frac{f_{\rm c}}{f_0}d_0$, thus effectively forming an \emph{S-ULA}, as illustrated in Fig. \ref{Fig:EquivalentMoveAntenna}.
	This adjustment provides flexible control of the beam-width (by expanding the effective array aperture) as well as null steering, which allows to effectively suppress the received signal power at Eves.
	\item The second term $\frac{\Delta_{f_n}r}{f_0\sin\theta} $ in $t_n$ is contributed by frequency offsets imposed on individual antennas.
	Since $\Delta_{f_n}$ is relatively small w.r.t. the carrier frequency, the position of each antenna is adjusted by only a small fraction of the wavelength on top of the S-ULA, as illustrated in Fig. \ref{Fig:EquivalentMoveAntenna}.
	Notably, these displacements can vary across different antennas, hence effectively forming a \emph{sparse non-uniform linear array} which enables more flexible beam control for PLS enhancement.
	In addition, it can be shown that the effective displacement $\frac{\Delta_{f_n}r}{f_0\sin\theta}$ (also known as FDAs \cite{wang2015frequency}) can form \emph{range-dependent} beams, hence more effectively reducing the received signal power at Eves.
	This approach can even achieve null steering when Bob and the Eves are located at the same angle.
\end{itemize}
\vspace{-10pt}
\subsection{Problem Formulation}
Based on the above discussions on the implementation issues, we formulate a practical optimization problem in this subsection, which aims to maximize the achievable secrecy rate of Bob by jointly optimizing the digital beamformer $\mathbf{w}_{\rm B}$, the carrier frequency $f_{\rm c}$ and the per-antenna FIV $\boldsymbol{\Delta}_f$, subject to the transmit power constraint for Alice and the frequency control constraints.
Specifically, the optimization problem can be formulated as
\begin{subequations}
	\begin{align}
	&({\rm P1}) \max_{\{f_{\rm c}, \boldsymbol{\Delta}_f, \mathbf{w}_{\rm B}\}} R_{{\rm S}}(\mathbf{w}_{\rm B}, \boldsymbol{\Delta}_f, f_{\rm c}) \\
	&~~~~~~~~~~~\text {s.t.}~\|\mathbf{w}_{\rm B}\|^2 \leq P_{\max}, \label{subeq:power constraint} \\
	&~~~~~~~~~~~~~~~~f_0 \le f_{\rm c} \le f_{\rm H},\label{subeq:carrier constraint} \\
	&~~~~~~~~~~~~~~~~0 \leq \Delta_{f_n} \le \Delta_{f_{\max}}, ~\forall n,\label{subeq:frequency constraint}\\
	&~~~~~~~~~~~~~~~~|\Delta_{f_i}-\Delta_{f_j}|\big(|\Delta_{f_i}\!-\!\Delta_{f_j}|\!-\!B\big) \ge 0, \forall i,j,\label{subeq:Filtering constraint}
	\end{align}
\end{subequations}
where $B$ denotes the bandwidth of the baseband signal $s_{\rm B}$
Herein, \eqref{subeq:power constraint} represents the transmit power constraint at Alice with $P_{\max}$ denoting its maximum transmit power.
In addition, \eqref{subeq:carrier constraint} and \eqref{subeq:frequency constraint} represent the constraints on the carrier frequency and the per-antenna frequency increments, respectively.
Finally, the constraint \eqref{subeq:Filtering constraint} ensures successfully mitigating the time-varying components ($2\pi \Delta_{f_n}t$).
In this paper, narrow-band signaling is considered, and thus we have $\Delta_{{f}_{\max}} \gg B$.
%Hence, to ensure fair comparison with MAs, the maximum carrier frequency of the FSA is chosen such that it corresponds to the maximum movable region of MAs, i.e., $\frac{f_{\rm H}}{f_0} = \frac{L}{D}$ where $D = Nd_0$ denotes the array aperture of the FPA.
\section{Single-BOB-Single-Eve Scenario}
\label{sec:beam width control}
In this section, to facilitate performance analysis, we characterize the beam control performance of proposed FSA under \emph{uniform} frequency increment over each antenna, while non-uniform frequency offsets will be discussed in Section \ref{Sec:solving optimization problem}.
To shed useful insights, we consider the following typical case.
\begin{itemize}
	\item A uniform frequency offset is considered.
	Mathematically, the FIV is given by {\small$\boldsymbol{\Delta}_f = [\Delta_f, 2\Delta_f,\cdots, N\Delta_f]^T$}.
	\item There is a single Eve and the maximum ratio transmission (MRT) beamformer (i.e.,  $\mathbf{w}_{\rm B} = \sqrt{P}_{\rm max}\mathbf{a}(f_{\rm c}, \boldsymbol{\Delta}_f,\theta_{\rm B},r_{\rm B})$) is adopted for data transmissions to Bob \cite{liu2025physical}.
	\item Eve is located within the main-lobe of the FPA, i.e., $\theta_{\rm E}\in (\theta_{\rm B}-\frac{2}{N},\theta_{\rm B}+\frac{2}{N})$, where $\theta_{\rm B}$ and $\theta_{\rm E}$ denote the AoAs of Bob and Eve, respectively.
	As such, we have $\Delta_{\theta} \triangleq \theta_{\rm B}-\theta_{\rm E} \in [-\frac{2}{N},\frac{2}{N}]$.
	This represents a challenging scenario where Eve can effectively wiretap confidential information.	
\end{itemize}
\begin{remark}[General case]
	\emph{Although the assumption of uniform frequency offsets reduces DoFs in the beam control, it allows for tractable analysis of PLS performance.
	The general case with non-uniform frequency offsets and multiple cooperative Eves is more complicated for performance analysis; as such, we solve the corresponding optimization problem in Section \ref{Sec:solving optimization problem}.
	In addition, for scenarios where Eve lies outside the main lobe of the beam steered towards Bob, i.e., $\theta_{\rm E} \notin (\theta_{\rm B} - \frac{2}{N}, \theta_{\rm B} + \frac{2}{N})$, effective eavesdropping cannot be achieved at Eve.
	In this scenario, it is not necessary to adjust the carrier frequency and frequency offsets to enhance PLS for Bob.
	Correspondingly, for MAs, it is not necessary to move antennas' positions.
	Hence, we only consider the challenging case $\theta_{\rm E}\in (\theta_{\rm B}-\frac{2}{N},\theta_{\rm B}+\frac{2}{N})$ in this paper.}
\end{remark}

Based on the above, (P1) reduces to
\begin{subequations}
	\begin{align}
	&({\rm P3}) \max_{\{f_{\rm c}, \boldsymbol{\Delta}_f\}} R_{{\rm FSA}} \triangleq \notag\\ &\bigg[\!\log_2\!\!\Big(\frac{\rho_{\rm B} + \sigma_0^2}{\rho_{\rm E}|\mathbf{a}^H(f_{\rm c},\! \boldsymbol{\Delta}_f,\!\theta_{\rm B},\!r_{\rm B})\mathbf{a}(f_{\rm c},\! \boldsymbol{\Delta}_f,\!\theta_{\rm E},\!r_{\rm E})|^2 \!\!+\!\! \sigma_0^2}\Big)\!\!\bigg]^{\!\!+}\label{eq:rewritten achievable secrecy rate} \\
	&~~~~~~~~~~\text {s.t.}~\eqref{subeq:carrier constraint}, \eqref{subeq:frequency constraint}, \eqref{subeq:Filtering constraint},\notag
	\end{align}
\end{subequations}
where $\rho_{\rm B} = P_{\rm \max}N\gamma_{\rm c}^2|g_{{\rm B},f_0}|^2$ and $\rho_{\rm E} = P_{\rm \max}N\gamma_{\rm c}^2|g_{{\rm E},f_0}|^2$ denote the effective received signal power at Bob and Eve, respectively.
It is observed from \eqref{eq:rewritten achievable secrecy rate} that the achievable secrecy rate of Bob is jointly determined by the channel correlation between Bob and Eve (i.e., {\small$G(f_{\rm c}, \Delta_{f})\triangleq|\mathbf{a}^H(f_{\rm c}, \boldsymbol{\Delta}_f,\theta_{\rm B},r_{\rm B})\mathbf{a}(f_{\rm c}, \boldsymbol{\Delta}_f,\theta_{\rm E},r_{\rm E})|$}) and the attenuation factor term $\gamma_{\rm c}$.

Note that the channel correlation $G(f_{\rm c}, \Delta_{f})$ and channel-gain attenuation factor $\gamma_{\rm c}$ are intricately coupled, rendering the analysis highly non-trivial.
To address this issue, we first consider a \emph{secrecy-guaranteed} problem with \emph{null-steering constraint}. Then we further characterize the relationship between secrecy-guaranteed problem and the original (P1).
Specifically, the secrecy-guaranteed problem can be formulated as follows
\begin{subequations}
	\begin{align}
		({\rm P4}) &\max_{\{f_{\rm c}, \boldsymbol{\Delta}_f\}} ~ R_{{\rm FSA}}\notag \\
		&~~~\text {s.t.}~\!|\mathbf{a}^H(f_{\rm c}, \boldsymbol{\Delta}_f,\!\theta_{\rm B},\!r_{\rm B})\mathbf{a}(f_{\rm c}, \boldsymbol{\Delta}_f,\theta_{\rm E},r_{\rm E})| = 0, \label{eq: null-steering constraint}\\
		&~~~~~~~0 \le \Delta_f \le \Delta_{f_{\max}},\label{eq:offset constraint 1}\\	
		&~~~~~~~\Delta_f(\Delta_f - B) \ge 0,\label{eq:offset constraint 2}\\
		&~~~~~~~f_0 \le f_{\rm c} \le f_{\rm H},\label{eq:carrier constraint}
	\end{align}
\end{subequations}
where \eqref{eq: null-steering constraint} represents the null-steering constraint.
As such, when (P4) is feasible, the objective function $R_{{\rm FSA}}$ reduces to $R_{{\rm FSA}}^{\rm (null)} = \log_2\Big(1 + \frac{ P_{\rm \max}N\gamma_{\rm c}^2|g_{{\rm B},f_0}|^2}{\sigma_0^2}\Big)$, which monotonically decreases with the carrier frequency.
Thus, the {optimal} solution to (P4) is to {minimize} the carrier frequency $f_{\rm c}$, while achieving null-steering over Eve.

\begin{lemma}
	\label{lemma: relationship between (P3) and (P4)}
	\rm{If the optimal solution $\{\Delta_{f}^{\ast}, f_{\rm c}^{\ast}\}$ to (P4) satisfies $f_{\rm c}^{\ast} = f_0$, then it is also an optimal solution to (P3).}
\end{lemma}
\begin{proof}
	Let $R_{{\rm FSA}}^{\dagger}$ and $R_{{\rm FSA}}^{\ast}$ denote the optimal values for (P3) and (P4), respectively.
	Since (P3) is a relaxed problem of (P4), we have $R_{{\rm FSA}}^{\dagger} \ge R_{{\rm FSA}}^{\ast}$.
	In addition, the optimal value for (P4) with $f_{\rm c}^{\ast} = f_0$ is given by
	\begin{equation}
	R_{{\rm FSA}}^{\ast}  = \log_2\Big(1 + {P_{\rm B}}/{\sigma_0^2}\Big) \ge R_{{\rm FSA}}^{\dagger},
	\label{eq:optimal achievable rate}
	\end{equation} 
	where $P_{\rm B} = P_{\rm \max}N|g_{{\rm B},f_0}|^2$.
	Thus, we have $R_{{\rm FSA}}^{\dagger} = R_{{\rm FSA}}^{\ast}$ and the proof is completed.
\end{proof}

{\bf Lemma \ref{lemma: relationship between (P3) and (P4)}} indicates that when only imposing frequency offsets to achieve null-steering over Eve, we can obtain the maximum  secrecy rate for Problem (P3).
However, when the optimal carrier frequency for (P4) satisfies $f_{\rm c}^{\ast} > f_0$, it is a \emph{suboptimal} solution to (P3) in general.

To meet the null-steering constraint \eqref{eq: null-steering constraint} in (P4), a closed-form expression for the Eve-and-Bob channel correlation is presented below.
\begin{lemma}[Channel correlation]
	\label{lemma:Closed-form channel correlation}
	\rm{Given the locations of Bob and Eve $(\theta_{\rm B}, r_{\rm B})$ and $(\theta_{\rm E}, r_{\rm E})$, the channel correlation between Bob and Eve can be obtained in the following form,}
	\begin{equation}
	G(f_{\rm c}, \Delta_{f}) = \bigg|\frac{\sin\big({N\pi}(\frac{1}{2}\frac{f_{\rm c}}{f_0}\Delta_{\theta} - \frac{\Delta_f\Delta r}{c})\big)}{N\sin\big(\frac{\pi}{2}\frac{f_{\rm c}}{f_0}\Delta_{\theta} - \frac{\pi\Delta_f\Delta r}{c}\big)}\bigg|,
	\label{eq:closed-form channel correlation}
	\end{equation}
	\rm{where $\Delta_{r} \triangleq r_{\rm B}-r_{\rm E} $ denotes the angle difference between Bob and Eve.} 
\end{lemma}
\begin{proof}
	For the channel correlation between Bob and Eve, we have
	{\small\begin{equation}
		\begin{aligned}
		G(f_{\rm c}, &\Delta_{f})=\frac{1}{N} \Big|\sum_{n=1}^{N} e^{\jmath\frac{2\pi}{c}f_{\rm c}\delta_n d_0\Delta_{\theta}-\jmath\frac{2\pi}{c}n\Delta_f\Delta_{r}}\Big|\\
		&=\Big|\frac{e^{\jmath\frac{2\pi}{c}(f_{\rm c}d_0\Delta_{\theta}-\Delta_f\Delta_{r})}(1-e^{\jmath\frac{2\pi}{c}N(f_{\rm c}d_0\Delta_{\theta}-\Delta_f\Delta_{r})})}{N(1-e^{\jmath\frac{2\pi}{c}(f_{\rm c}d_0\Delta_{\theta}-\Delta_f\Delta_{r})})}\Big|\\
		&=\Big|\frac{e^{-\jmath\frac{\pi}{c}N(f_{\rm c}d_0\Delta_{\theta}-\Delta_f\Delta_{r})}-e^{\jmath\frac{\pi}{c}N(f_{\rm c}d_0\Delta_{\theta}-\Delta_f\Delta_{r})}}{N(e^{-\jmath\frac{\pi}{c}{f_{\rm c}d_0\Delta_{\theta}-\Delta_f \Delta_{r}}}-e^{\jmath\frac{\pi}{c}(f_{\rm c}d_0\Delta_{\theta}-\Delta_f \Delta_{r})})}\Big|.	
		\end{aligned}
		\label{eq:semi closed form}
		\end{equation}}By applying the Euler's formula to \eqref{eq:semi closed form}, we can directly arrive at \eqref{eq:closed-form channel correlation}, thus completing the proof.
\end{proof}

Based on {\bf Lemma \ref{lemma:Closed-form channel correlation}}, we can achieve null steering at Eve by setting ${N\pi}(\frac{1}{2}\frac{f_{\rm c}}{f_0}\Delta_{\theta} - \frac{\Delta_f\Delta r}{c}) = k\pi$, with $k \in \{\mathbb{Z}~\cap~ \mod(k,N) \neq 0\}$, which is simplified as
\begin{equation}
	\frac{f_{\rm c}}{f_0}\Delta_{\theta} = \frac{2k}{N}+\frac{2\Delta_{r}\Delta_f}{c}.
	\label{eq:null-point-line}
\end{equation}

According to \eqref{eq:null-point-line}, the optimal solution to (P4) can be obtained for the following three cases: 
1) $\Delta_{\theta} = 0, \Delta_{r} \neq 0$, i.e., Eve and Bob are located at the same angle but different ranges; 2) $\Delta_{r} = 0, \Delta_{\theta} \neq 0$, i.e., Eve and Bob are located at the same range but different angles; and 3) $\Delta_{\theta} \neq 0, \Delta_{r} \neq 0$.
In the following subsections, we analyze the optimal design of the carrier frequency and frequency offsets for the proposed FSA in different cases, respectively.
%, with the main results summarized in Fig. \ref{Fig:selection}.

In addition, for comparison, we define the performance gap between the proposed FSA and benchmark schemes as follows.
\begin{definition}[Performance gap\footnote{Compared with MA and FPA, the proposed FSA generally needs a wider system bandwidth due to the operation of frequency hopping. For ease of comparison, we assume the same signal bandwidth for the three antenna configurations such that $ \sigma_{{\rm MA}}^2 = \sigma_{{\rm FPA}}^2 = \sigma_0^2 = BN_0$, where $N_0$ denotes the noise power spectral density, leading to the respective achievable rates in \eqref{eq:MA rate} and \eqref{eq:FPA rate}.}]
	\emph{The achievable secrecy rates of MAs and FPAs are given by}
	\begin{equation}
		R_{{\rm MA}} \!=\! \bigg[\log_2\Big(\frac{P_{\rm B} + \sigma_0^2}{P_{\rm E}|\mathbf{a}^H_{\rm MA}(\mathbf{x}, \theta_{\rm B})\mathbf{a}_{\rm MA}(\mathbf{x}, \theta_{\rm E})|^2 + \sigma_0^2}\Big)\!\bigg]^+\!\!,
		\label{eq:MA rate}
	\end{equation}
	\begin{equation}
		R_{{\rm FPA}} \!=\! \bigg[\log_2\Big(\frac{P_{\rm B} + \sigma_0^2}{P_{\rm E}|\mathbf{a}^H_{\rm FPA}(\theta_{\rm B})\mathbf{a}_{\rm FPA}(\theta_{\rm E})|^2 +  \sigma_0^2}\Big)\!\bigg]^+,
		\label{eq:FPA rate}
	\end{equation}
	\emph{where $P_{\rm E} = P_{\rm \max}N|g_{{\rm E},f_0}|^2$ denotes the received signal power at Eve under the basic frequency $f_0$ and $\mathbf{a}_{\rm MA}(\mathbf{x}, \theta) = \frac{1}{\sqrt{N}}\big[ e^{\jmath{2\pi}\frac{f_{\rm c}[\mathbf{x}]_1\sin\theta}{c}}, \cdots, e^{\jmath{2\pi}\frac{f_{\rm c}[\mathbf{x}]_N\sin\theta}{c}}\big]^T$ is the array response vector of MA with $\mathbf{x}$ representing the antenna position vector.
	In addition, $\mathbf{a}_{\rm FPA}(\theta) = \frac{1}{\sqrt{N}}\big[ e^{\jmath{2\pi}\frac{f_{\rm c}\delta_1d_0\sin\theta}{c}}, \cdots, e^{\jmath{2\pi}\frac{f_{\rm c}\delta_Nd_0\sin\theta}{c}}\big]^T$ denotes the array response vector of FPA.
	As such, the performance gaps of FSA over the two benchmark arrays are given by
	\begin{equation}
		\Delta R_{{\rm MA}}  =  R_{{\rm S}, {\rm FSA}} - R_{{\rm S}, {\rm MA}},~\Delta R_{{\rm FPA}} = R_{{\rm S}, {\rm FSA}} - R_{{\rm S}, {\rm FPA}}.
		\notag
	\end{equation}}
\end{definition}
%In particular, when $\Delta R_{{\rm S, MA}} (\Delta R_{{\rm S, FPA}}) > 0$, the performance gap is referred as performance \emph{gain} of the proposed FSA, while we refer {\small$|\Delta R_{{\rm S, MA}}|~ (|\Delta R_{{\rm S, FPA}}|)$} as performance \emph{loss} when {\small$\Delta R_{{\rm S, MA}}~ (\Delta R_{{\rm S, FPA}}) < 0$}.

\vspace{-10pt}
\subsection{First Case: $\Delta_{\theta} = 0, \Delta_{r} \neq 0$}
\label{sec:case 1}
For this case, the null-steering condition in \eqref{eq:null-point-line} reduces to
\begin{equation}
	{\Delta_f}= -\frac{kc}{N\Delta_{r}}.
	\label{eq:null-point first case}
\end{equation}

Typically, we have ${\Delta_f} \ge \frac{kc}{N\Delta_{r}} > B$.
For example, when Alice is equipped with $N = 13$ antennas and the range difference between Bob and Eve is $\Delta_{r} \in [10,100]$~m, the minimal frequency offset in \eqref{eq:null-point first case} can be obtained as $\min_k -\frac{kc}{N\Delta_{r}} = \frac{c}{N|\Delta_{r}|} \in [0.23, 2.3]$ MHz.
Considering that the bandwidth of a narrow-band subcarrier signal in the 5G frequency range 1 (FR1) is 15--120 kHz, we have ${\Delta_f} \ge \frac{kc}{N\Delta_{r}} > B$ and thus the constraint \eqref{eq:offset constraint 2} can be omitted in the analysis of this section.

Moreover, it is observed that switching the carrier frequency $f_{\rm c}$ cannot control the null steering of the MRT beam, when Bob and Eve are located at the same angle (i.e., $\Delta_{\theta} = 0$).
This can be intuitively understood, as switching the carrier frequency virtually creates a larger-aperture sparse array, leading to a smaller beam-width in the angular domain, while it offers limited control over beam properties in the range domain.
%\footnote{If the adjustable carrier frequency is sufficiently large such that the equivalent aperture expansion (see Section \ref{sec:Physical Interpretation: A Movable Antenna Perspective}) places users in the near-field region of the virtual array, the beam-focusing effect emerges, and adjusting the carrier frequency can then help reduce power leakage to Eves.}
Fortunately, applying small uniform frequency offsets across the antennas (i.e., the so-called FDA) can form range-dependent beams, hence enabling null steering in the range domain.
As such, we obtain the optimal frequency offset $\Delta_f$ and the carrier frequency $f_{\rm c}$ as follows.

\begin{figure}[t]
	\centering
	\vspace{-14pt}
	\includegraphics[width=0.365\textwidth]{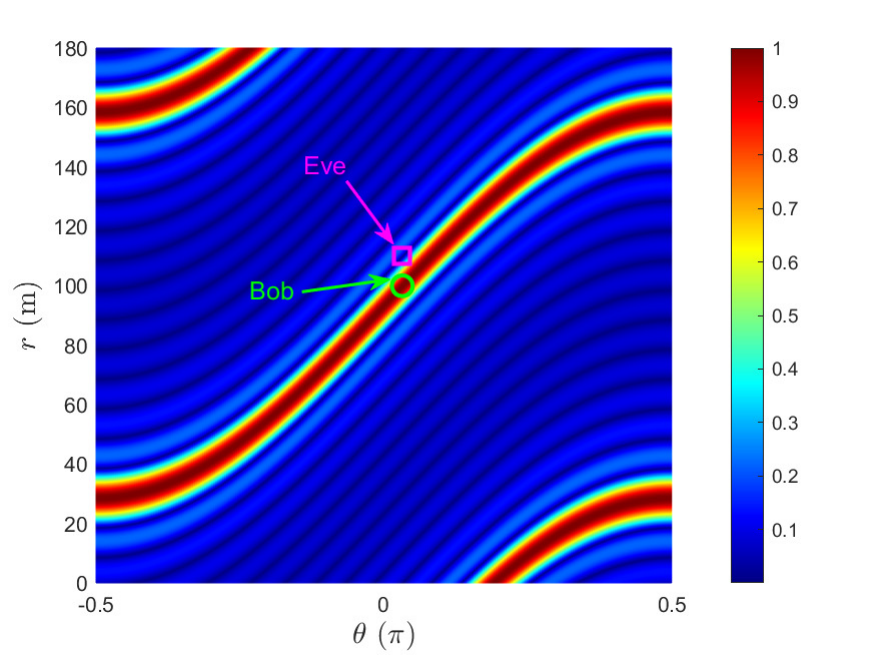}
	\caption{The channel correlation in the range domain when $\Delta_{\theta}$ = 0. The system parameters are given by $\theta_{\rm B} = \theta_{\rm E} = 0.1$, $r_{\rm B} = 100$ m and $r_{\rm E} = 110$ m (i.e., $\Delta_{r} = 10$ m).
		In addition, Alice is equipped with $N=13$ antennas and operates at $f_{\rm c} = f_0 = 3.5$ GHz and the frequency offset is $\Delta_f = \frac{c}{N\Delta_{r}} = 2.3$ MHz.} 
	\label{Fig:NullSreering Range}
	\vspace{-12pt}
\end{figure}
\begin{figure}[t]
	\vspace{-5pt}
	\centering
	\includegraphics[width=0.3\textwidth]{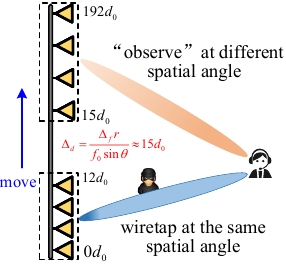}
	\caption{Physical interpretation of Case 1.} 
	\label{Fig:Physical interpretation of Case 1}
	\vspace{-10pt}
\end{figure}
\begin{proposition}[Optimal solution to (P4) for Case 1]
	\label{proposition:Optimal solution to (P4) for case 1}
	\rm{When $\Delta_{\theta} = 0, \Delta_{r} \neq 0$, Problem (P4) is feasible if and only if $\Delta_{f_{\max}} \ge \frac{c}{N|\Delta_{r}|}$ and the optimal solution is}
	\begin{equation}
	\Delta_f^{\ast} =
	\frac{c}{N|\Delta_{r}|},~ f_{\rm c}^{\ast} = f_0.
	\end{equation}
\end{proposition}
\begin{proof}
	The result can be easily obtained from \eqref{eq:null-point first case} and we omit the details for brevity.
\end{proof}
\begin{corollary}[Optimal solution to (P3) for Case 1]
	\label{lemma:Optimal solution to (P3)}
	\rm{When $\Delta_{\theta} = 0, \Delta_{r} \neq 0$, the optimal solution to (P3) is given by}
	\begin{equation}
		\Delta_f^{\ast} = \min\Big\{ \Delta_{f_{\max}},\frac{c}{N|\Delta_{r}|} \Big\},~ f_{\rm c}^{\ast} = f_0.
	\end{equation}
\end{corollary}
\begin{proof}
	When $\Delta_{\theta} = 0$, the channel correlation reduces to
	\begin{equation}
	G(f_{\rm c}, \Delta_{f}) = \bigg|\frac{\sin\big({N\pi}\frac{\Delta_f\Delta r}{c}\big)}{N\sin\big(\frac{\pi\Delta_f\Delta r}{c}\big)}\bigg| \triangleq \widetilde{G}(\Delta_{f}),
	\label{eq:reduced closed-form channel correlation}
	\end{equation}
	which is independent of the carrier frequency $f_{\rm c}$.
	Hence, the optimal $f_{\rm c}$ is given by $f_{\rm c}^{\ast} = f_0$.
	Then, the achievable secrecy rate $R_{\rm FSA}$ in \eqref{eq:rewritten achievable secrecy rate} monotonically decreases with the channel correlation $\widetilde{G}(\Delta_{f})$.
	Based on Lemma \ref{lemma:Optimal solution to (P3)}, if $\Delta_{f_{\max}} \ge \frac{kc}{N|\Delta_{r}|}$, $\Delta_f = \frac{kc}{N|\Delta_{r}|}$ is an optimal solution to (P4), which is also the optimal solution to (P3) according to Lemma \ref{lemma: relationship between (P3) and (P4)}.
	
	If $\Delta_{f_{\max}} < \frac{c}{N|\Delta_{r}|}$, we have $-\frac{1}{N}<\frac{\Delta_f\Delta r}{c} < \frac{1}{N}$.
	As such, the channel correlation $\widetilde{G}(\Delta_{f})$ monotonically decreases with $\Delta_f$.
	Hence, the optimal frequency offset can be selected as $\Delta_f^{\ast} = \Delta_{f_{\max}}$ to minimize $\widetilde{G}(\Delta_{f})$, which arrives at the optimal solution to (P4), thus completing the proof.
\end{proof}

\begin{remark}[Physical interpretation of {\bf Proposition \ref{proposition:Optimal solution to (P4) for case 1}}]
	\emph{{\bf Proposition \ref{proposition:Optimal solution to (P4) for case 1}} indicates that even when Bob and Eve share the same angle, the confidential signal can bypass Eve by imposing frequency increments on individual antennas as illustrated in Fig. \ref{Fig:NullSreering Range},  which can be intuitively understood as follows.
	Taking the parameters in Fig. \ref{Fig:NullSreering Range} as an example, when $\Delta_{f} = 2.3$ MHz, the position of the $n$-th antenna is equivalent to being shifted upwards by $\frac{\Delta_{f_n}r}{f_0\sin\theta} = \frac{n\Delta_{f}r}{f_0\sin\theta} \approx 15d_0$ (see Section \ref{sec:Physical Interpretation: A Movable Antenna Perspective}), as illustrated in Fig. \ref{Fig:Physical interpretation of Case 1}.
	As such, the FSA is effectively shifted by a considerable distance, which makes the observed angles from Bob and Eve become different.
	As a result, the beam steered towards Bob can effectively “bypass” Eve.
}
\end{remark}

However, in far-field communications, if Bob and Eve are located at the same angle, neither FPAs nor MAs (assuming the movable region is insufficiently large to make Eves located within the near-field region) can effectively reduce the channel correlation between Bob and Eve, i.e., $|\mathbf{a}^H(\theta_{\rm B},r_{\rm B})\mathbf{a}(\theta_{\rm E},r_{\rm E})| = 1$. As such, the achievable secrecy rates for FPAs and MAs are
\begin{equation}
\begin{aligned}
R_{{\rm S}, {\rm MA}}^{(1)} = R_{{\rm S}, {\rm FPA}}^{(1)} = \bigg[\log_2\Big(\frac{P_{\rm B} + \sigma_0^2}{P_{\rm E}+ \sigma_0^2}\Big)\bigg]^+.
\end{aligned}
\label{eq:performance fixed array}
\end{equation}

\begin{corollary}[FSAs versus MAs for Case 1]
	\label{corollary:case 1 copared with FPA}
	\rm{For Case 1, the secrecy-rate performance \emph{gain} of the proposed FSA over the FPA and MA is given by
	\begin{equation}
		\begin{aligned}
			\Delta &R_{{\rm MA}}^{(1)} = \Delta R_{{\rm FPA}}^{(1)} =  R_{{\rm S}, {\rm FSA}}^{\ast} - R_{{\rm S}, {\rm MA}}^{(1)}\\
			&~~~~~~= \min\Big\{\log_2\big(1 + \frac{P_{\rm E}}{\sigma_0^2}\big),	
			\log_2\big(1 + \frac{P_{\rm B}}{\sigma_0^2}\big)\Big\} > 0.
		\end{aligned}
		\label{eq:performance gain}
	\end{equation}}
\end{corollary}
\begin{proof}
	Combining \eqref{eq:performance fixed array} and \eqref{eq:optimal achievable rate} directly leads to \eqref{eq:performance gain} and thus completing the proof.
\end{proof}

{\bf{Corollary}} \ref{corollary:case 1 copared with FPA} indicates that the proposed FSA always outperforms MA and FPA, when Bob and Eve are located at the same angle.
Moreover, it is observed from \eqref{eq:performance gain} that as the received SNRs at Eve and Bob increase, the secrecy rate gain achieved by the proposed FSA becomes more significant.

\vspace{-5pt}
\subsection{Second Case: $\Delta_{r} = 0, \Delta_{\theta} \neq 0$}
\label{sec:second case}
For this case, the null-steering condition of the array gain in \eqref{eq:null-point-line} reduce to ${f_{\rm c}} = \frac{2k}{N\Delta_{\theta}}{f_0}$.
It is observed that when Eve and Bob are located at the same range, adjusting the frequency offsets of all antennas cannot effectively realize null-steering over Eve, which, however, can be achieved by switching the carrier frequency.
The optimal solution to (P4) for Case 2 is presented below.
\begin{proposition}[Optimal solution to (P4) for Case 2]
	\label{lemma:Optimal solution to (P4) for case 2}
	\rm{When $\Delta_{\theta} \neq 0, \Delta_{r} = 0$, Problem (P4) is feasible if and only if ${f_{\rm H}} \ge \frac{2}{N|\Delta_{\theta}|}{f_0}$. In this case, its optimal solution and the corresponding objective value are respectively given by}
	\begin{equation}
		\Delta_f^{\ast} = 0,~ f_{\rm c}^{\ast} = \frac{2}{N|\Delta_{\theta}|}{f_0}
	\end{equation} 
	\begin{equation}
		R_{{\rm FSA}}^{\ast}  = \log_2\Big(1 + \frac{P_{\rm B}N^2|\Delta_{\theta}|^2}{4\sigma_0^2}\Big).
		\label{eq:optimal value fpr case 2}
	\end{equation}
\end{proposition}
\begin{proof}
	The result can be easily obtained and we omit the detailed proof for brevity.
\end{proof}

To compare the performance upper bound with MAs, we assume that the carrier frequency can vary without an upper limit, which corresponds to an unbounded movable region for MAs.
As such, MAs can always achieve null-steering over Eve \cite{zhu2023movableNullSteering}, for which its achievable secrecy rate is given by $R_{{\rm S}, {\rm MA}}^{(2)} = \log_2\Big(1+\frac{P_{\rm B}}{\sigma_0^2}\Big)$.
\begin{corollary}[FSAs versus MAs for Case 2]
	\label{corollary:case 2 copared with FPA}
	\rm{When $\Delta_f = 0$ and $f_{\rm c} = \frac{2}{N|\Delta_{\theta}|}{f_0}$ given in {\bf Proposition \ref{lemma:Optimal solution to (P4) for case 2}}, the secrecy-rate performance loss of the proposed FSA over MA is given by
		\begin{equation}
		\begin{aligned}
		\Delta R_{{\rm MA}}^{(2)} =  R_{{\rm S}, {\rm FSA}}^{\ast} - R_{{\rm S}, {\rm MA}}^{(2)}
		= \log_2\Big(\frac{\frac{P_{\rm B}N^2|\Delta_{\theta}|^2}{4} + \sigma_0^2}{P_{\rm B} + \sigma_0^2}\Big) \le 0.
		\end{aligned}
		\label{eq:performance loss}
		\end{equation}}
\end{corollary}
\begin{proof}
	Combining $R_{{\rm S}, {\rm MA}}^{(2)}$ and \eqref{eq:optimal value fpr case 2} directly leads to \eqref{eq:performance loss}, thus completing the proof.
\end{proof}

According to {\bf{Corollary} \ref{corollary:case 2 copared with FPA}}, in the high signal-to-noise ratio (SNR) regime, the secrecy-rate performance gain of MAs over FSAs can be approximated as $\Delta R_{{\rm MA}}^{(2)} \approx 2\log_2\big(\frac{P_{\rm B}N|\Delta_{\theta}|}{2}\big)$, which increases with the angle difference between Bob and Eve in a logarithmic order.
This can be intuitively understood, as a smaller angular separation between Eve and Bob requires the FSA to form a narrower beam to achieve null-steering over Eve, which in turn requires a higher carrier frequency.
The increased carrier frequency leads to a larger channel-gain attenuation factor, thereby enlarging the performance gap between the proposed FSA and MA.

Although FSAs exhibit a performance gap as compared to MAs, it still achieves significant performance improvement over FPAs.
In particular, the secrecy rate of the FPA is
\begin{equation}
	R_{{\rm S}, {\rm FPA}}^{(2)} = \log_2\Big(\frac{P_{\rm B} + \sigma_0^2}{I(\Delta_{\theta})P_{\rm B}+\sigma_0^2}\Big),
\label{eq:rate fixed array}
\end{equation}
where $I(\Delta_{\theta}) = |\mathbf{a}^H(\theta_{\rm B},r_{\rm B})\mathbf{a}(\theta_{\rm E},r_{\rm E})| = \Big|\frac{\sin\big({\frac{1}{2}N\pi}\Delta_{\theta}\big)}{N\sin\big(\frac{\pi}{2}\Delta_{\theta}\big)}\Big|$.
Next, we present the condition under which the proposed FSA outperforms FPA.
\begin{figure*}[h]
	\vspace{-12pt}
	\centering
	\captionsetup[subfloat]{labelfont=rm, format=plain,labelformat=empty}	
	\subfloat[{{\small (a) Geometric explanation.}}]{\includegraphics[width=0.5\columnwidth]{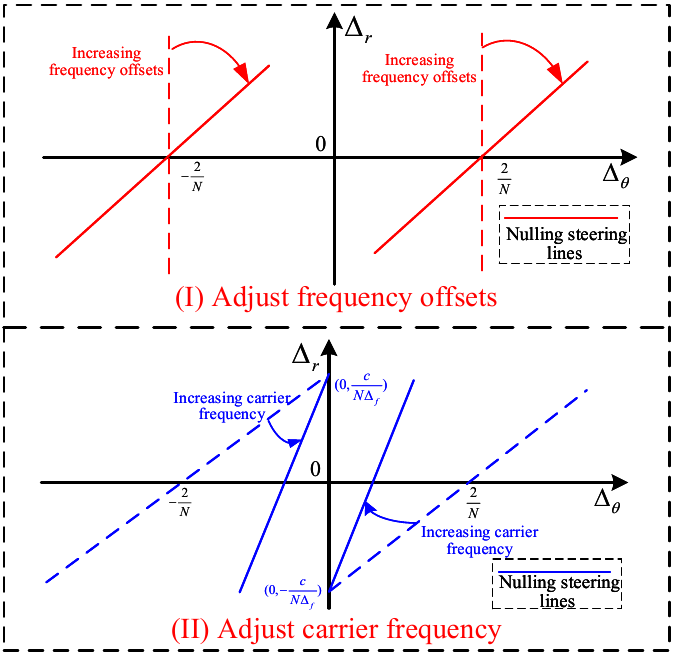}}
	\hfil
	\subfloat[{\small(b) Ilustration of {\bf Proposition \ref{lemma:case 3 1}}.}] {\includegraphics[width=0.78\columnwidth]{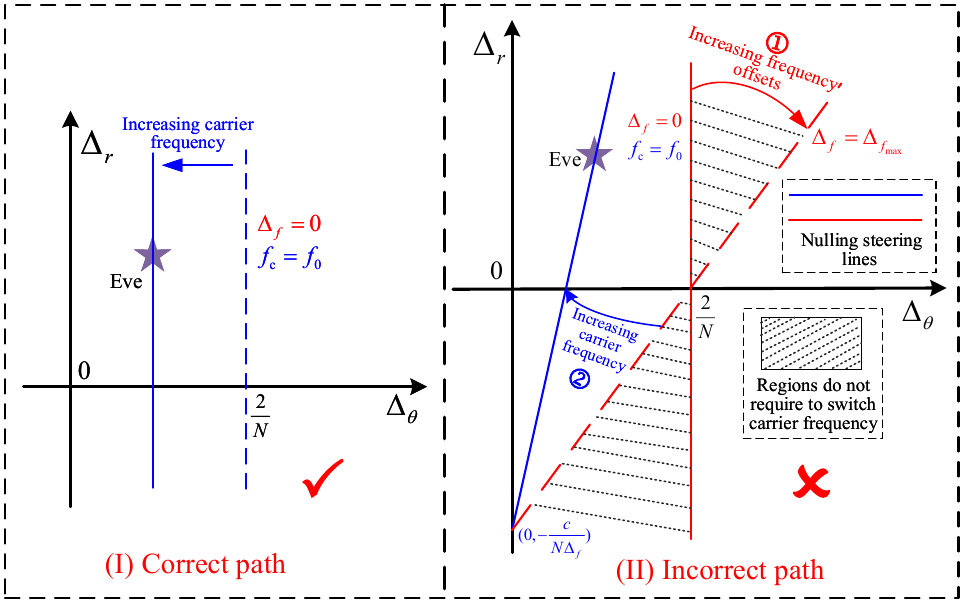}}
	\hfil
	\subfloat[{\small (c) Ilustration of {\bf Proposition \ref{lemma:case 3 2}}.}]{\includegraphics[width=0.5\columnwidth]{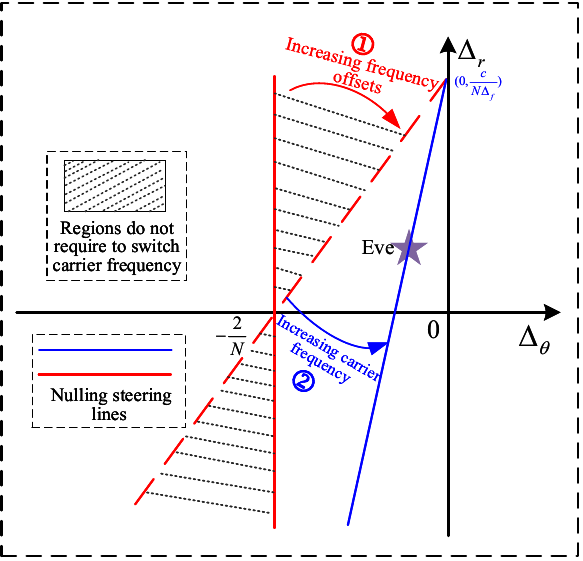}}
	\caption{{Physical interpretation of the selection of carrier frequency and the corresponding offset.}}
	\label{fig:physical}
	\vspace{-15pt}
\end{figure*}
% \begin{figure*}[t]
% 	\vspace{-14pt}
% 	\centering
% 	\subfloat[Carrier frequency switching condition]{
% 		\includegraphics[width=0.45\textwidth]{./Figures/Condition.eps}}
% 	\subfloat[Achievable secrecy rate versus carrier frequency]{
% 		\includegraphics[width=0.45\textwidth]{./Figures/NullSteering_fc.eps}}
% 	\caption{The system parameters are given by $r_{\rm B} = r_{\rm E} = 100$ m (i.e., $\Delta_{r} = 0$ m). In addition, Alice is equipped with $N=13$ antennas, while the transmit power and noise power are $P_{\max} = 30$ dBm and $\sigma^2 = -80$ dBm, respectively.} 
% 	\label{Fig:condition}
% 	\vspace{-5pt}
% \end{figure*}
\begin{corollary}[FSAs versus FPA for Case 2]
	\label{lemma:condition}
	\rm{When ${f_{\rm c}}^{\ast} = \frac{2}{N\Delta_{\theta}}{f_0}$ and $\Delta_{r} = 0$ (i.e., $P_{\rm B} = P_{\rm E}$), the secrecy-rate performance \emph{gain} of the proposed FSA over FPA is given by
		\begin{equation}
			\begin{aligned}
				&\Delta R_{{\rm FPA}}^{(2)} =  [R_{{\rm S}, {\rm FSA}}^{\ast} - R_{{\rm S}, {\rm FPA}}^{(2)}]^{+}\\
				&= \Bigg[\log_2\Bigg(\frac{\Big(\frac{P_{\rm B}N^2|\Delta_{\theta}|^2}{4} + \sigma_0^2\Big)\Big(I(\Delta_{\theta})P_{\rm B}+\sigma_0^2\Big)}{\sigma_0^2(P_{\rm B} + \sigma_0^2)}\Bigg)\Bigg]^{+}.
			\end{aligned}
		\end{equation}}
\end{corollary}
\begin{proof}
	Switching the carrier frequency as that in {\bf{Proposition}} \ref{lemma:Optimal solution to (P4) for case 2} is necessary when $\Delta R_{{\rm FPA}}^{(2)} > 0$ and the result can be easily obtained.
\end{proof}

Based on {\bf{Corollary} \ref{corollary:case 2 copared with FPA}}, let $\Delta R_{{\rm FPA}}^{(2)} \ge 0$, which can be simplified as {\small$I(\Delta_{\theta}) \ge I_{\rm th}(\Delta_{\theta}) \triangleq \frac{(1-\frac{N^2\Delta_{\theta}^2}{4})\sigma_0^2 }{\frac{N^2\Delta_{\theta}^2}{4}P_{\rm B}+\sigma_0^2}$}.
As such, if the angle difference $\Delta_\theta$ between Eve and Bob satisfies $I(\Delta_{\theta}) \ge I_{\rm th}(\Delta_{\theta})$, the proposed FSA outperforms FPA in terms of the achievable secrecy rate and the corresponding optimal carrier frequency should be switched as that given in {\bf{Proposition} \ref{lemma:Optimal solution to (P4) for case 2}}.
%Otherwise, the FSA keeps the carrier frequency unchanged. 
%Indeed, in the high-SNR regime, the condition is readily satisfied.
This implies that switching the carrier frequency to enable null steering towards Eve is beneficial for enhancing the achievable secrecy rate as compared to FPA.
It is worth noting that if the design objective is to ensure no information leakage to Eves rather than maximizing the secrecy rate, the carrier frequency should still be tuned following the parameter settings specified in {\bf Proposition \ref{lemma:Optimal solution to (P4) for case 2}}.
%In Fig. \ref{Fig:condition}(a), we plot $I(\Delta_{\theta})$ and $I_{\rm th}(\Delta_{\theta})$ versus $\Delta_{\theta}$. It is observed that when the angular separation between Eve and Bob lies within the main lobe of the MRT-based beam, the condition always holds.
%In addition, in Fig. \ref{Fig:condition}(b), we plot the achievable secrecy rate versus carrier frequency. It is shown that switching the carrier frequency can significantly enhance the achievable secrecy rate as compared to FAs with basic carrier frequency $f_0 = 30$ GHz.
%Although switching the carrier frequency as that in Lemma \ref{lemma:Optimal solution to (P4) for case 2} to achieve null steering toward Eve may not be the optimal solution to (P3), it closely approaches the performance of the optimal design.
\vspace{-8pt}
\subsection{Third Case: $\Delta_{\theta} \neq 0, \Delta_{r} \neq 0$}
For the general case, we demonstrate that null steering over Eve can always be achieved.
\emph{Interestingly}, it is shown that when $\Delta_{\theta}\Delta_{r} > 0$, either switching the carrier frequency or imposing a common frequency offset can realize null steering, while using both methods does not provide additional performance gains.
In contrast, when $\Delta_{\theta}\Delta_{r} < 0$, the common frequency offset should be tuned to its largest value for achieving the optimal solution to (P4), if the available frequency offset is not sufficiently large.
\begin{proposition}[Optimal solution to (P4) when $\Delta_{\theta}\Delta_{r} > 0$]
	\label{lemma:case 3 1}
	\rm{When $\Delta_{\theta}\Delta_{r} > 0$, the optimal solution to (P4) is given by}
	{\begin{equation}
		\label{eq:case 3 1}
		\left\{\begin{matrix}	
		f_{\rm c}^{\ast} = \frac{2f_0}{N|\Delta_{\theta}|},\Delta_f^{\ast} = 0, \text{if}~\Delta_{f_{\max}} \le \frac{c}{2\Delta_{r}}(\Delta_{\theta}+\frac{2\Delta_{\theta}}{N|\Delta_{\theta}|}),\\
		\!\!\!\!\!\!\!f_{\rm c}^{\ast} = f_0,~~~~\Delta_f^{\ast} = \frac{c}{2\Delta_{r}}(\Delta_{\theta}+\frac{2\Delta_{\theta}}{N|\Delta_{\theta}|}), ~\text{otherwise}.
		\end{matrix}\right. 
		\end{equation}}
\end{proposition}
\begin{proof}
	First, if $\Delta_{f_{\max}} \ge \frac{c}{2\Delta_{r}}(\Delta_{\theta}+\frac{2\Delta_{\theta}}{N|\Delta_{\theta}|})$, it is easy to verify that $\{\Delta_f^{\ast} = \frac{c}{2\Delta_{r}}(\Delta_{\theta}+\frac{2\Delta_{\theta}}{N|\Delta_{\theta}|}), f_{\rm c}^{\ast} = f_0\}$ is the optimal solution to (P4).
	Next, if $\Delta_{f_{\max}} < \frac{c}{2\Delta_{r}}(\Delta_{\theta}+\frac{2\Delta_{\theta}}{N|\Delta_{\theta}|})$ with $\Delta_\theta > 0 $ and $\Delta_r > 0$, we have
	\begin{equation}
	0 \le \frac{2\Delta_{r}\Delta_f}{c} < \Delta_{\theta} + \frac{2}{N}.
	\label{eq:fanwei}
	\end{equation} 
	Since the null-steering condition is $\frac{f_{\rm c}}{f_0}\Delta_{\theta} = \frac{2k}{N}+\frac{2\Delta_{r}\Delta_f}{c}$ (see \eqref{eq:null-point-line}), $\frac{2k}{N}+\frac{2\Delta_{r}\Delta_f}{c}$ should satisfy
	\begin{equation}
	\frac{2k}{N}+\frac{2\Delta_{r}\Delta_f}{c} \ge \Delta_{\theta}.
	\label{eq:conflict}
	\end{equation} 
	If $k \le -1$, based on \eqref{eq:fanwei}, we can obtain $\frac{2k}{N} + \frac{2\Delta_{r}\Delta_f}{c} < \Delta_{\theta}$, which contradicts \eqref{eq:conflict}, thereby implying that $k > 0$ must hold.
	As such, to minimize $f_{\rm c}$, we should choose $\Delta_f^{\ast} = 0$ with $k=1$ and $f_{\rm c}^{\ast} = \frac{2f_0}{N\Delta_{\theta}}$.
	Similarly, when $\Delta_\theta < 0 $ and $\Delta_r < 0$, the optimal solution to (P4) can be obtained as $\Delta_f^{\ast} = 0$ and $f_{\rm c}^{\ast} = -\frac{2f_0}{N\Delta_{\theta}}$.
	Combining the results of the above two cases leads to \eqref{eq:case 3 1}, hence completing the proof.
\end{proof}

{\bf{Proposition} \ref{lemma:case 3 1}} indicates that in scenarios where $\Delta_{\theta}\Delta_{r} > 0$, a single frequency control method (i.e., switching carrier frequency or imposing frequency offset) is sufficient to achieve null steering over Eve. The corresponding performance analysis can be found in Sections \ref{sec:case 1} and \ref{sec:second case}, respectively.

\begin{proposition}[Optimal solution to (P4) when $\Delta_{\theta}\Delta_{r} < 0$]
	\label{lemma:case 3 2}
	\rm{If $\Delta_{\theta} \Delta_{r} < 0$, the optimal solution to (P4) is given by}
	{\small\begin{equation}
		\notag
		\left\{\begin{matrix}
		\!\!\!\Delta_f^{\ast} \!=\! \frac{c}{2\Delta_{r}}(\Delta_{\theta}\!-\!\frac{2\Delta_{\theta}}{N|\Delta_{\theta}|}),~f_{\rm c}^{\ast} \!=\! f_0,~ \text{if}~\Delta_{f_{\max}} \!\ge\! \frac{c}{2\Delta_{r}}(\Delta_{\theta}\!-\!\frac{2\Delta_{\theta}}{N|\Delta_{\theta}|}),\!\\ 
		\Delta_f^{\ast} = \Delta_{f_{\max}},~~~~~~~~~~~~~ f_{\rm c}^{\ast} = \frac{2f_0}{N|\Delta_{\theta}|} + \frac{2\Delta_{r}\Delta_{f_{\max}}f_0}{c\Delta_{\theta}},~\!{\rm{otherwise}}.
		\end{matrix}\right. 
		\end{equation}}
\end{proposition}
\begin{proof}
	First if $\Delta_{f_{\max}} \ge \frac{c}{2\Delta_{r}}(\Delta_{\theta}-\frac{2\Delta_{\theta}}{N|\Delta_{\theta}|})$, it is easy to verify that $\{\Delta_f^{\ast} = \frac{c}{2\Delta_{r}}(\Delta_{\theta}-\frac{2\Delta_{\theta}}{N|\Delta_{\theta}|}), f_{\rm c}^{\ast} = f_0\}$ is the optimal solution to (P4).
	Next, if $\Delta_{f_{\max}} < \frac{c}{2\Delta_{r}}(\Delta_{\theta}-\frac{2\Delta_{\theta}}{N|\Delta_{\theta}|})$ with $\Delta_\theta > 0$ and $\Delta_r < 0$, we have
	\begin{equation}
		\Delta_{\theta} - \frac{2}{N} < \frac{2\Delta_{r}\Delta_f}{c} \le 0.
		\label{eq:new condition}
	\end{equation} 
	For the null-steering condition $\frac{f_{\rm c}}{f_0}\Delta_{\theta} = \frac{2k}{N}+\frac{2\Delta_{r}\Delta_f}{c}$, given $k = 1$ and \eqref{eq:new condition}, we have $\frac{2k}{N} + \frac{2\Delta_{r}\Delta_f}{c} > \Delta_{\theta}$, which is sufficient to achieve null-steering over Eve.
	 As such, to minimize $f_{\rm c}$, we should choose $\Delta_f^{\ast} = \Delta_{f_{\max}}$ with $k=1$ and $f_{\rm c}^{\ast} = \frac{2f_0}{N\Delta_{\theta}} + \frac{2\Delta_{r}\Delta_{f_{\max}}f_0}{c\Delta_{\theta}}$.
	 Similarly, when $\Delta_\theta < 0 $ and $\Delta_r > 0$, we can obtain the optimal solution to (P4) as $\Delta_f^{\ast} = \Delta_{f_{\max}}$ and $f_{\rm c}^{\ast} = -\frac{2f_0}{N\Delta_{\theta}} + \frac{2\Delta_{r}\Delta_{f_{\max}}f_0}{c\Delta_{\theta}}$.
	 Thus, combining these results of the two cases leads to the desired results.
\end{proof}

{\bf{Proposition} \ref{lemma:case 3 2}} indicates that in scenarios where $\Delta_{\theta}\Delta_{r} < 0$, if the frequency offset is not sufficiently large, the frequency offset should be maximized to achieve optimal solution to (P4).
As compared to Case 2 in Section \ref{sec:second case}, the required switching carrier frequency is smaller, indicating that the introduction of a maximum frequency offset further enhances PLS performance.
\begin{remark}[Geometric explanation for {\bf Proposition \ref{lemma:case 3 1} } and {\bf \ref{lemma:case 3 2}}]
	\emph{Since we only consider the case where $\theta_{\rm E}\in (\theta_{\rm B}-\frac{2}{N},\theta_{\rm B}+\frac{2}{N})$, i.e., $\Delta_{\theta} \in [-2/N,2/N]$, only two null-steering lines in \eqref{eq:null-point-line} ($k=1$ and $k=-1$) should be considered as illustrated in Figs. \ref{fig:physical}(b) and \ref{fig:physical}(c). 
	We first rewrite $\frac{f_{\rm c}}{f_0}\Delta_{\theta} = \frac{2k}{N}+\frac{2\Delta_{r}\Delta_f}{c}$ in \eqref{eq:null-point-line} as $\Delta_r = \frac{c}{2\Delta_ff_0}f_{\rm c}\Delta_{\theta} - \frac{kc}{N\Delta_f}$.
	Two key observations are made as follows.
	\begin{itemize}
		\item First, increasing the carrier frequency \textit{increases} the slope of the null-steering lines, whereas enlarging the frequency offsets \textit{reduces} their slope.
		\item Second, adjusting the frequency offset make the null-steering lines rotate around the point ($\pm N/2$, $0$) as illustrated in Fig. \ref{fig:physical}(a-I), while switching the carrier frequency makes the null-steering lines rotate around their intercept with the $\Delta$-axis (i.e., the points $(0, \frac{kc}{N\Delta_f}), k = \pm 1$) as shown in Fig. \ref{fig:physical}(a-II).
	\end{itemize}
	For {\bf Proposition \ref{lemma:case 3 1}}, we take $\Delta_r > 0, \Delta_{\theta} > 0$ as an example.
	As shown in Fig. \ref{fig:physical}(b), when the maximal frequency offsets fails to rotate the null-steering lines toward Eve’s position, the slope of these lines should be \textit{increased rather than decreased}.
	In this case, the frequency offset should be set to zero instead of being enlarged to its maximum value, otherwise a higher carrier frequency is required to shift the null-steering lines toward Eve’s location (see Fig. \ref{fig:physical}(b-II)).
	Instead, for the case in {\bf Proposition \ref{lemma:case 3 2}} (taking $\Delta_r > 0, \Delta_{\theta} < 0$ as an example), we should \textit{reduce} the slope of the null steering line as shown in Fig. \ref{fig:physical}(c), which is achieved by increasing the frequency offset.
	Although this imposing maximal frequency offsets across individual antennas is insufficient to achieve null steering over Eve, a subsequent increase in the carrier frequency can rotate the null steering line around the point ($\frac{c}{N\Delta_f}$,0) towards Eve \textit{in a relay manner}, as shown in Fig. \ref{fig:physical}(c).
}
\end{remark}

\vspace{-10pt}
\section{Proposed Solution to Problem (P1)}
\label{Sec:solving optimization problem}
In this section, we propose an efficient algorithm to solve Problem (P1) by using the BCD and PGA methods. Specifically, we decompose (P1) into three subproblems, which
\emph{alternately} optimize the transmit beamformer, FIV and carrier frequency, respectively.
\vspace{-5pt}
\subsection{Optimize $\mathbf{w}_{\rm B}(t)$ given $\boldsymbol{\Delta}_f$ and $f_{\rm c}$}
Given any feasible $\boldsymbol{\Delta}_f$ and $f_{\rm c}$, Problem (P1) reduces to 
\begin{subequations}
	\begin{align}
	&({\rm P1.1})~~ \max_{\mathbf{w}_{\rm B}} \frac{\mathbf{w}_{\rm B}^H\mathbf{A}\mathbf{w}_{\rm B} + 1}{\mathbf{w}_{\rm B}^H\mathbf{B}\mathbf{w}_{\rm B}+ 1} \\
	&~~~~~~~~~~~\text {s.t.}~\|\mathbf{w}_{\rm B}\|^2 \leq P_{\max},\label{subeq:power constraint 2} 
	\end{align}
\end{subequations}
where $\mathbf{A} = \frac{1}{\sigma^2}\mathbf{h}_{\rm B}(f_{\rm c}, \boldsymbol{\Delta}_f)\mathbf{h}^H_{\rm B}(f_{\rm c}, \boldsymbol{\Delta}_f) $ and $\mathbf{B} = \frac{1}{\sigma_0^2}\sum_{m=1}^{M}\mathbf{h}_{{\rm E},m}(f_{\rm c}, \boldsymbol{\Delta}_f)\mathbf{h}^H_{{\rm E},m}(f_{\rm c}, \boldsymbol{\Delta}_f)$.
The objective function in (P2) is the well-known generalized Rayleigh quotient, for which its optimal solution is given by
\begin{equation}
	\mathbf{w_{\rm B}^{\ast}} = \sqrt{P_{\rm max}}\mathbf{c}_{\rm max}.
	\label{eq:optimal beamformer}
\end{equation}
Herein, $\mathbf{c}_{\rm max}$ denotes the eigenvector w.r.t. the largest eigenvalue of the matrix $ \big(\mathbf{B}+\frac{1}{P_{\rm max}}\mathbf{I}_{N}\big)^{-1}\big(\mathbf{A}+\frac{1}{P_{\rm max}}\mathbf{I}_{N}\big)$.

\subsection{Optimize $f_{\rm c}$ given $\mathbf{w}_{\rm B}(t)$ and $\boldsymbol{\Delta}_f$}
Given any feasible $\mathbf{w}_{\rm B}(t)$ and $\boldsymbol{\Delta}_f$, the objective function $R_{{\rm S}}(\mathbf{w}_{\rm B}, \boldsymbol{\Delta}_f, f_{\rm c})$ in (P1) is a one-dimensional (1D) function w.r.t. $f_{\rm c}$, which can be effectively solved via a 1D linear search. 
Specifically, let $N_f$ denote the number of samples, and then the search space can be constructed as $\mathcal{F} = \big\{ f_n = f_0 + \frac{n-1}{N_f - 1}\left(f_H - f_0\right),n = 1,2,\ldots,N_f 
\big\}$. 
Herein, the search step size is $\frac{n-1}{N_f - 1}(f_{\rm H}-f_0)$. As such, the optimal carrier frequency is given by
\begin{equation}
\label{eq:1D search}
	f_{\rm c}^{\ast} = \arg\max\limits_{\mathcal{F}} R_{{\rm S}}(\mathbf{w}_{\rm B}, \boldsymbol{\Delta}_f, f_{\rm c}).
\end{equation}

\vspace{-5pt}
\subsection{Optimize $\boldsymbol{\Delta}_f$ given $\mathbf{w}_{\rm B}(t)$ and $f_{\rm c}$}
\label{sec:optimize f}
%When optimizing the FIV $\boldsymbol{\Delta}_f$, the obtained optimal digital beamformer $\mathbf{w_{\rm B}^{\ast}}$ is utilized.
To facilitate gradient computation, let $\mathbf{w_{\rm B}^{\ast}} = \mathbf{u} + \jmath \mathbf{v}$, where $\mathbf{u}$ and $\mathbf{v}$ are the real and imaginary parts of $\mathbf{w_{\rm B}^{\ast}}$, respectively.
For convenience, we rewrite the channel vector for Eve $m$ as $\mathbf{h}_{{\rm E},m}(f_{\rm c}, \boldsymbol{\Delta}_f) = h_m (\mathbf{x}_m + \jmath \mathbf{y}_m)$,
where $h_m = \frac{f_0}{f_c}g_{{\rm E},f_0}^{(m)} $ denotes the path gain for Eve $m$.
In addition, $\mathbf{x}_m = [x_{m,1},\cdots,x_{m,n},\cdots, x_{m,N}]^{T} $ and $\mathbf{y}_m = [y_{m,1},\cdots,y_{m,n},\cdots, y_{m,N}]^{T} $, where
%\begin{equation}
%	x_{m,n} = \cos\Big(\frac{2\pi}{c}(f_{\rm c}\delta_n d_0\sin\theta_m-\Delta_{f_n} r_m) + 2\pi\Delta_{f_n} t\Big),
%\end{equation}
\begin{equation}
x_{m,n} = \cos\Big(\frac{2\pi}{c}(f_{\rm c}\delta_n d_0\sin\theta_m-\Delta_{f_n} r_m) \Big),
\end{equation}
\begin{equation}
	y_{m,n} = \sin\Big(\frac{2\pi}{c}(f_{\rm c}\delta_n d_0\sin\theta_m-\Delta_{f_n} r_m)\Big).
\end{equation}
%\begin{equation}
%	y_{m,n} = \sin\Big(\frac{2\pi}{c}(f_{\rm c}\delta_n d_0\sin\theta_m-\Delta_{f_n} r_m) + 2\pi\Delta_{f_n} t\Big).
%\end{equation}

Similarly, the channel vector for Bob can be rewritten as $\mathbf{h}_{{\rm B}}(f_{\rm c}, \boldsymbol{\Delta}_f) \triangleq  (\mathbf{x}_0 + \jmath \mathbf{y}_0)h_{0} $, where $h_{{0}} = \frac{f_0}{f_c}g_{{\rm B},f_0} $ represents the effective channel gain of Bob. Moreover, $\mathbf{x}_0$ and $\mathbf{y}_0$ have similar expressions with $\mathbf{x}_m$ and $\mathbf{y}_m$, which are omitted for brevity.

Based on the above, given any feasible $\mathbf{w}_{\rm B}^{\ast}$ and $f_c$, the optimization problem (P1) reduces to
\begin{subequations}
	\begin{align}
	&({\rm P1.2})~ \max_{\boldsymbol{\Delta}_f} \widetilde{R}_{{\rm S}}(\boldsymbol{\Delta}_f) = \log_2\Big(1 + {\widetilde{\gamma}_0 g(\mathbf{x}_0, \mathbf{y}_0)} \Big) \\
	&~~~~~~~~~~~~~~~~~~~~-\log_2\Big(1 + {\sum_{m=1}^{M}\widetilde{\gamma}_m g(\mathbf{x}_m, \mathbf{y}_m)} \Big)\notag\\
	&~~~~~~~~~~~\text {s.t.}~\eqref{subeq:frequency constraint},\eqref{subeq:Filtering constraint},
	\end{align}
\end{subequations}
where $\widetilde{\gamma}_i = \frac{|h_0|^2}{\sigma_0^2}, i\in \mathcal{I} \triangleq \{0,1,\cdots,M\} $. Specifically, the function $g(\mathbf{x}_i, \mathbf{y}_i)$ is defined as $g(\mathbf{x}_i, \mathbf{y}_i) = \mathbf{x}_i^T\mathbf{C}\mathbf{x}_i + \mathbf{y}_i^T\mathbf{C}\mathbf{y}_i + 2\mathbf{x}_i^T\mathbf{D}\mathbf{y}_i$, where $\mathbf{C} \triangleq \mathbf{u}\mathbf{u}^T + \mathbf{v}\mathbf{v}^T$ and $\mathbf{D} \triangleq \mathbf{u}\mathbf{v}^T - \mathbf{v}\mathbf{u}^T$. 

(P1.2) is a non-convex optimization problem (P1.2) for which its suboptimal solution can be found by using the PGA method. Specifically, given the FIV $\boldsymbol{\Delta}_f^{(s)}$ in the $s$-th iteration, $\boldsymbol{\Delta}_f^{(s+1)}$ in the $(s+1)$-th iteration can be updated as 
\begin{subequations}
	\label{eq:update f}
	\begin{align}
		&\text{(Step. 1)}~~ \boldsymbol{t}^{(s+1)} =
		\boldsymbol{\Delta}_f^{(s)} + \mu\nabla_{\boldsymbol{\Delta}_f^{(s)}}\widetilde{R}_{{\rm S}}(\boldsymbol{\Delta}_f),\\
		&\text{(Step. 2)}~~
		\boldsymbol{\Delta}_f^{(s+1)} =\arg \min _{\boldsymbol{\Delta}_f}\big|\big|\boldsymbol{t}^{(s+1)}-\boldsymbol{\Delta}_f\big|\big|\\
		&~~~~~~~~~~~~~~~~~~~~~~~~\text{s.t.}~ \eqref{subeq:frequency constraint},\eqref{subeq:Filtering constraint},\notag
	\end{align}
\end{subequations}
where $\widetilde{R}_{{\rm S}}(\boldsymbol{\Delta}_f)$ represents the gradient of $\widetilde{R}_{{\rm S}}(\boldsymbol{\Delta}_f)$ w.r.t. $\boldsymbol{\Delta}_f$ and $\mu$ denotes the step size for gradient ascent, respectively.

For the non-convex problem (P1.3), it can be transformed into a mixed-integer quadratic programming (MIQP) problem by using the 1D clustering method, given by
\begin{subequations}
	\begin{align}
		 &\min_{\{\boldsymbol{\Delta}_f, \mathbf{u},\mathbf{c},\mathbf{Z}\}}\big|\big|\boldsymbol{\Delta}_f - \boldsymbol{t}^{(s+1)}\big|\big|_2^2
		 \\&~~~~\text{s.t.}~u_k \in\{0,1\}, z_{n,k}\in\{0,1\}.\forall k,n,\label{eq:binary constraint}
		 \\&~~~~~~~\sum_{k = 1}^K z_{n,k} = 1, \forall n\in\mathcal{N},\label{eq:only one actioveate}
		 \\&~~~~~~~z_{n,k} \le u_k,\forall n,k, \label{eq:only one cluster}
		 \\ &~~~~~~~|\Delta_{f_n} - c_k| \le C_{\rm large}(1 - z_{n,k}), \forall k,n,\label{eq:assign}
		 \\&~~~~~~~c_{k+1} \ge c_k + B - C_{\rm large}(2 - u_k - u_{k+1}),\forall k,\label{eq:gap}
		 \\&~~~~~~~0 \le {\Delta}_{f_n} \le \Delta_{f_{\max}},~0 \le c_k \le \Delta_{f_{\max}},\forall n,k,
	\end{align}
\end{subequations}where $K$ denotes the number of clusters and $C_{\rm large}$ represents a sufficiently large constant. 
Herein, $\mathbf{c} = [c_1,c_2,\cdots,c_K]^T$ is the cluster center vector, while $\mathbf{u} = [u_1,u_2,\cdots,u_K]^T$ represents cluster activation variables with $u_k = 1$ denoting the $k$-th cluster is activated.
In addition, $\mathbf{Z}$ denotes the binary assignment matrix, where $z_{n,k} = [\mathbf{Z}]_{n,k} = 1$ indicates that the $n$-th entry in $\boldsymbol{\Delta}_f$ to be optimized is assigned to cluster $k$.
The constraint \eqref{eq:binary constraint} restricts the assignment and activation variables to binary values, while the constraints \eqref{eq:only one actioveate} and \eqref{eq:only one cluster} ensure a entry in $\boldsymbol{\Delta}_f$ is assigned to only one cluster.
Moreover, if the $n$-th entry in $\boldsymbol{\Delta}_f$ is assigned to the $k$-th cluster, \eqref{eq:assign} ensures $\boldsymbol{\Delta}_{f_n} = c_k$, and the constraint \eqref{eq:gap} restricts a minimum separation of $B$ between any two consecutive active clusters.
The above MIQP problem can be effectively solved by using the branch-and-cut algorithm \cite{bienstock1996computational}.  
%However, when the $n$-th entry $[\boldsymbol{\Delta}_f^{(s+1)}]_n$ in $\boldsymbol{\Delta}_f^{(s+1)}$ doss not fall within the range $[0, \Delta_{f_{\max}}]$, $[\boldsymbol{\Delta}_f^{(s+1)}]_n$ should be recast by using the projection procedure, given by~\cite{hu2024secure}
%\begin{equation}
%	[\boldsymbol{\Delta}_f^{(s+1)}]_n = \max\Big\{0, \min\big\{\Delta_{f_{\max}}, [\boldsymbol{\Delta}_f^{(s+1)}]_n\big\}\Big\}.
%	\label{eq:update f 2}
%\end{equation}

Next, we obtain the closed-form gradient of $\widetilde{R}_{{\rm S}}(\boldsymbol{\Delta}_f)$ w.r.t. $\boldsymbol{\Delta}_f$ in closed form as follows.
\begin{proposition}
	\rm{Given any feasible beamformer $\mathbf{w_{\rm B}^{\ast}} = \mathbf{u} + \jmath \mathbf{v}$, the gradient of $\widetilde{R}_{{\rm S}}(\boldsymbol{\Delta}_f)$ w.r.t. $\boldsymbol{\Delta}_f$ is given by}
	{\small\begin{equation}
		\begin{aligned}
		&\nabla_{\boldsymbol{\Delta}_f}\widetilde{R}_{{\rm S}}(\boldsymbol{\Delta}_f) = \frac{\widetilde{\gamma}_0\Big(\mathbf{X}_0\big(2\mathbf{C}\mathbf{x}_0 + 2\mathbf{D}\mathbf{y}_0\big) + \mathbf{Y}_0\big(2\mathbf{C}\mathbf{y}_0 + 2\mathbf{D}^T\mathbf{x}_0\big)\Big)}{\ln(2)\big(1 + {\widetilde{\gamma}_0 g(\mathbf{x}_0, \mathbf{y}_0)}\big)}\\  &-\frac{\sum_{m=1}^{M}\widetilde{\gamma}_m\Big(\mathbf{X}_m\big(2\mathbf{C}\mathbf{x}_m + 2\mathbf{D}\mathbf{y}_m\big) + \mathbf{Y}_m\big(2\mathbf{C}\mathbf{y}_m + 2\mathbf{D}^T\mathbf{x}_m\big)\Big)}{\ln(2)\big(1 + {\sum_{m=1}^{M} \widetilde{\gamma}_m g(\mathbf{x}_m, \mathbf{y}_m)}\big)},
		\end{aligned}
		\label{eq:full GA}
		\end{equation}}{\rm where
%		{\small$\mathbf{X}_i = \text{diag}\Big\{\frac{2\pi}{c}r_i\sin\Big(\frac{2\pi}{c}(f_{\rm c}\delta_n d_0\sin\theta_m\!-\!\Delta_{f_n} r_m) + 2\pi\Delta_{f_n} t\Big)\Big\}$} and
%		{\small$\mathbf{Y}_i \!=\! \text{diag}\Big\{-\frac{2\pi}{c}r_i\cos\Big(\frac{2\pi}{c}(f_{\rm c}\delta_n d_0\sin\theta_m-\Delta_{f_n} r_m) + 2\pi\Delta_{f_n} t\Big)\Big\}$}	
		{\small$\mathbf{X}_i = \text{diag}\Big\{\frac{2\pi}{c}r_i\sin\Big(\frac{2\pi}{c}(f_{\rm c}\delta_n d_0\sin\theta_m\!-\!\Delta_{f_n} r_m)\Big)\Big\}$} and
		{\small$\mathbf{Y}_i \!=\! \text{diag}\Big\{-\frac{2\pi}{c}r_i\cos\Big(\frac{2\pi}{c}(f_{\rm c}\delta_n d_0\sin\theta_m-\Delta_{f_n} r_m) \Big)\Big\}, \forall i \in \mathcal{I}$}}.	
	\label{prpo:GA f}
\end{proposition}
\vspace{-1pt}
\begin{proof}
%	Please refer to \bf Appendix A.
	The proof is similar to that in \cite{hu2024secure} and we omit it for brevity.
\end{proof}

Based on {\bf Proposition \ref{prpo:GA f}}, the FIV $\boldsymbol{\Delta}_f$ can be iteratively updated via the PGA method which converges to a locally optimal solution.
\begin{remark}[Algorithm convergence and computational complexity]
	\emph{First, we consider the algorithm convergence. For solving (P1.2) via the PGA method, the value of the objective function $R_{{\rm S}}(\mathbf{w}_{\rm B}, \boldsymbol{\Delta}_f, f_{\rm c})$ is non-decreasing over iterations.
	In addition, the function $R_{{\rm S}}(\mathbf{w}_{\rm B}, \boldsymbol{\Delta}_f, f_{\rm c})$ under the power and frequency switching constraints is upper-bounded by $R_{{\rm S}}(\mathbf{w}_{\rm B}, \boldsymbol{\Delta}_f, f_{\rm c}) \le \log_2\big( 1 + \frac{ NP_{\rm max}|g_{{\rm B},f_0}|^2}{\sigma_0^2}\big)$, which ensures the convergence of the proposed optimization algorithm.
	Next, we analyze the computational complexity of the proposed algorithm which is determined by three parts: 1) eigenvalue decomposition and matrix inversion in obtaining the optimal beamformer; 2) the one-dimensional linear search in \eqref{eq:1D search}; and 3) calculating the gradients $\nabla_{\boldsymbol{\Delta}_f}\widetilde{R}_{{\rm S}}(\boldsymbol{\Delta}_f)$ and $\nabla_{f_{\rm c}}{\widehat{R}_{{\rm S}}(f_{\rm c})}$.
	First, the computational complexity of the eigenvalue decomposition and matrix inversion for obtaining the optimal beamformer $\mathbf{w}_{\rm B}$ is $ \mathcal{O}\big(N^3\big) $.
	Second, the computational complexity of the one-dimensional linear search for the optimal carrier frequency is given by $ \mathcal{O}\big(N_fMN\big) $.
	Third, for obtaining the gradient $\nabla_{\boldsymbol{\Delta}_f}\widetilde{R}_{{\rm S}}(\boldsymbol{\Delta}_f)$, the computational complexity is in the order of $ \mathcal{O}\big(MN\big) $ \cite{ranjan2023gradient}.
	Hence, the overall computational complexity is approximately {\small$\mathcal{O}\Big(T_{\rm BCD}\big(N^3+\big(N_f + T_{\rm GA}^{\rm FIV}\big)MN \big)\Big)$, where $T_{\rm BCD}$} and $T_{\rm GA}^{\rm FIV}$ denote the number of iterations in the BCD and PGA methods, respectively.}
\end{remark}

\vspace{-8pt}
\subsection{Discussion: CSI Acquisition for Eve}
\label{sec:Extension to CSI Acquisition for Eve}
The CSI acquisition for Eve is critical for the design of carrier frequency and frequency offsets across individual antenna in Problem (P1), which can be practically acquired by two main approaches, namely, radar echo-based sensing and RF leakage-based detection.
For the first approach, Alice can actively transmit probing signals and estimate Eve’s CSI from the reflected echoes by using e.g., subspace-based localization techniques \cite{su2023sensing}.
Alternatively, Eve’s location can be inferred by exploiting the RF leakage from the local oscillator (LO) of a passive receiver \cite{mukherjee2012detecting}.
Specifically, owing to hardware constraints, a small portion of the LO signal leaks through the RF front end into the environment, even when Eve remains passive.
Alice can detect such leakage via energy detection or advanced statistical methods such as the generalized likelihood ratio test.
Once power leakage is detected, spatially distributed antenna observations combined with subspace-based techniques can be used to estimate the location of Eve.

In practice, the CSI acquisition for Eves may not be perfect, which thus necessities robust beamforming design for enhancing PLS performance.
Specifically, the actual channel of the $m$-th Eve can be modelled as
\begin{equation}
	\mathbf{h}_{{\rm E},m} = \hat{\mathbf{h}}_{{\rm E},m} + \Delta\mathbf{h}_{{\rm E},m}, m \in \mathcal{M}, 
	\label{eq:estimation error}
\end{equation}
where $\hat{\mathbf{h}}_{{\rm E},m}$ is the estimated Eve channel used for beamforming design and $\Delta\mathbf{h}_{{\rm E},m}$ denotes the channel estimation error, satisfying $\Delta\mathbf{h}_{{\rm E},m} \in \mathcal{U}_m \triangleq \{||\Delta\mathbf{h}_{{\rm E},m}||_2 \le \xi_{m} \}$.
Herein, $\xi_{m}$ represents the radii of the uncertainty regions.
%Specifically, the imperfect CSI of the $m$-th Eve is modeled by uncertainty regions $\mathcal{R}_m = \{r_m \mid r_m = \hat{r}_m + \Delta r_m\}$ and $\mathcal{A}_m = \{\theta_m \mid \theta_m = \hat{\theta}_m + \Delta \theta_m\}$, where $|\Delta r_m| \le \xi_{r,m}$ and $|\Delta \theta_m| \le \xi_{\theta,m}$, with $\xi_{r,m}$ and $\xi_{\theta,m}$ representing the radii of the range and angular uncertainty regions, respectively.
Then, we aim to maximize the minimal secrecy rate of the uncertainty channel of the $M$ Eves, which can be formulated as
\begin{subequations}
	\begin{align}
		&({\rm P5}) \max_{\{f_{\rm c}, \boldsymbol{\Delta}_f, \mathbf{w}_{\rm B}(t)\}} \Big\{R_{{\rm B}} - \max_{\mathcal{U}_m} R_{{\rm E}}\Big\}\\
		&~~~~~~~~~~~~~\text {s.t.}~\eqref{subeq:power constraint}, \eqref{subeq:carrier constraint}, \eqref{subeq:frequency constraint}. \notag
	\end{align}
\end{subequations} 
By introducing auxiliary variables $s$, (P5) can be rewritten as
\begin{subequations}
	\begin{align}
		&({\rm P6}) \max_{\{f_{\rm c}, \boldsymbol{\Delta}_f, \mathbf{w}_{\rm B},s\}} R_{{\rm B}}- s,\\
		&~~~~~~~~~~~~\text {s.t.}~R_{{\rm E}} \le s, \forall \Delta\mathbf{h}_{{\rm E},m} \in \mathcal{U}_m  \label{eq:robust constraint},\\
		&~~~~~~~~~~~~~~~~\eqref{subeq:power constraint}, \eqref{subeq:carrier constraint}, \eqref{subeq:frequency constraint}. \notag
	\end{align}
\end{subequations}

\begin{figure}[t]
	\centering
	\vspace{-14pt}
	\includegraphics[width=0.38\textwidth]{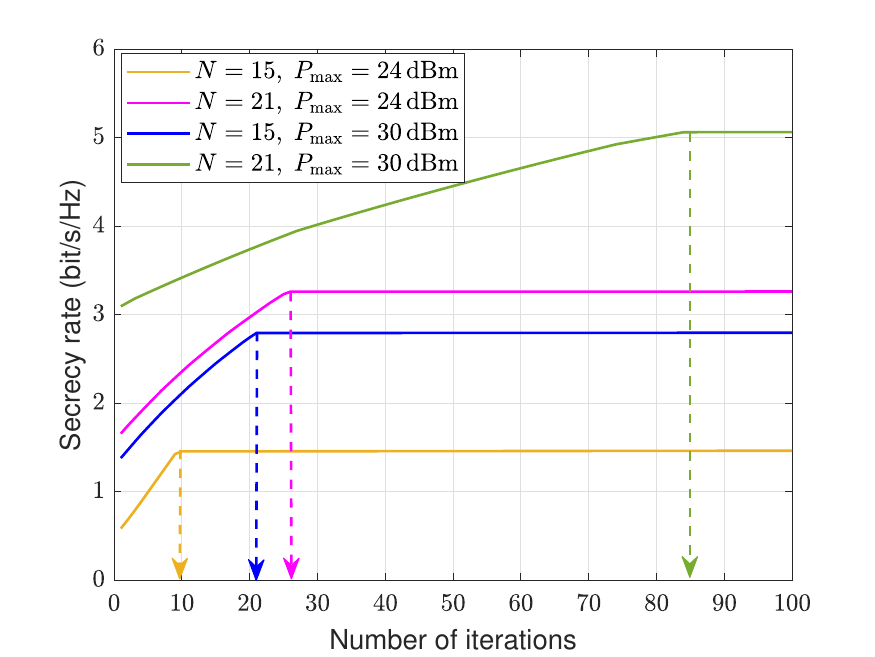}
	\caption{Convergence behavior of proposed algorithm} 
	\label{Fig:convergence}
	\vspace{-16pt}
\end{figure}
As compared to Problem (P1), (P6) introduces an additional non-convex constraint \eqref{eq:robust constraint}.
The projection gradient descent method used for the subproblems of (P1) is no longer applicable, and the location uncertainty of Eves results in \emph{infinitely} many non-convex constraints in \eqref{eq:robust constraint} due to the continuous uncertainty sets, which makes (P6) intractable.
One typical solution approach is by using the S-Procedure and sign-definiteness to transform the infinitely many constraints into a finite number of constraints (see, e.g., \cite{zhou2020framework}), which allows the use of the BCD method for solving Problem (P6)with details omitted for brevity.
The impact of imperfect CSI on the secrecy-rate performance of the proposed FSA will be evaluated through numerical results in Section~\ref{sec:numerical results} (see Fig.~\ref{Fig:bound}).
%In the following, due to space limitations, we provide only a brief outline of the solution approach.
%In particular, when optimizing $\{\mathbf{w}_{\rm B},s\}$, with other variables fixed, the SCA technique can be applied to reformulate the subproblem into a tractable convex problem, as discussed in \cite{lyu2025robust}.
%Moreover, when optimizing $\boldsymbol{\Delta}_f$, the simulated annealing (SA) improved particle swarm optimization (SA-PSO) algorithm can be employed to obtain a satisfactory solution \cite{yu2022novel}.
%Finally, the carrier frequency $f_{\rm c}$ can be optimized via a one-dimensional linear search.
\vspace{-12pt}
\section{Numerical Results}
\label{sec:numerical results}
\subsection{System Parameters}
The system parameters are set as follows, unless specified otherwise.
We consider that Alice operates at basic frequency band $f_0 = 3.5$ GHz, who transmits confidential information to a Bob located at $(0^{\circ}, 40~\text{m})$.
In addition, $M = 2$ Eves are considered in this section, which are located at $(0.71^{\circ}, 41~\text{m})$ and $(-0.94^{\circ}, 40~\text{m})$.
Moreover, the maximum frequency offset at each antenna is $\Delta_{{f}_{\max}} = 2.5$ MHz, while the maximum carrier frequency that can be switched is $f_{\rm H} = 7$ GHz.
Finally, and the maximal transmit power and the bandwidth of baseband signals are set as $P_{\max} = 30$ dBm and $B = 0.5$ kHz.
For performance comparison, we consider the following benchmark schemes: 1) FPAs; 2) MAs; and 3) FDAs.

\begin{figure}[t]
	\centering
	\vspace{-14pt}
	\includegraphics[width=0.38\textwidth]{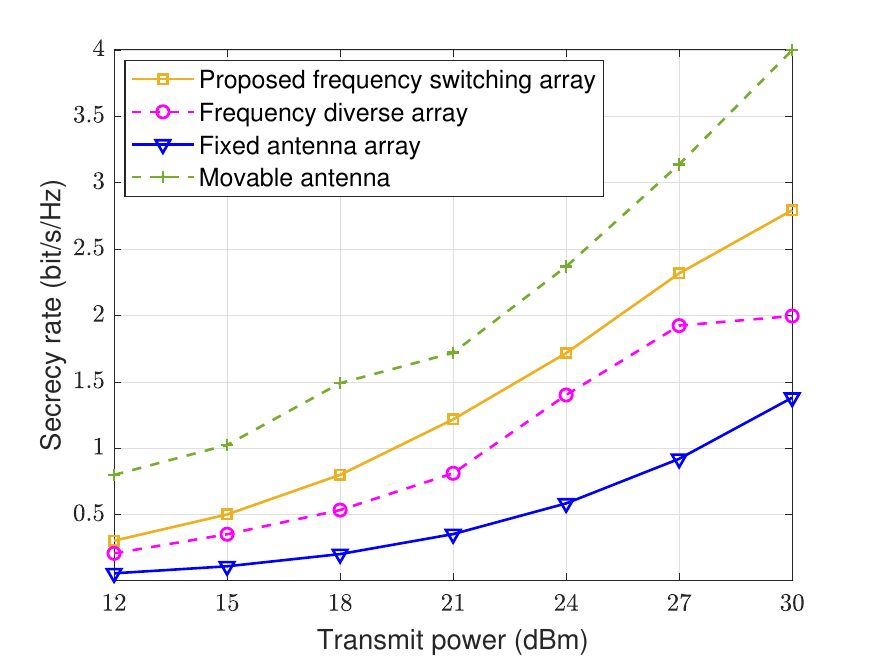}
	\caption{Secrecy rate versus transmit power} 
	\label{Fig:Power}
	\vspace{-5pt}
\end{figure}
\begin{figure}[t]
	\centering
	\vspace{-10pt}
	\includegraphics[width=0.38\textwidth]{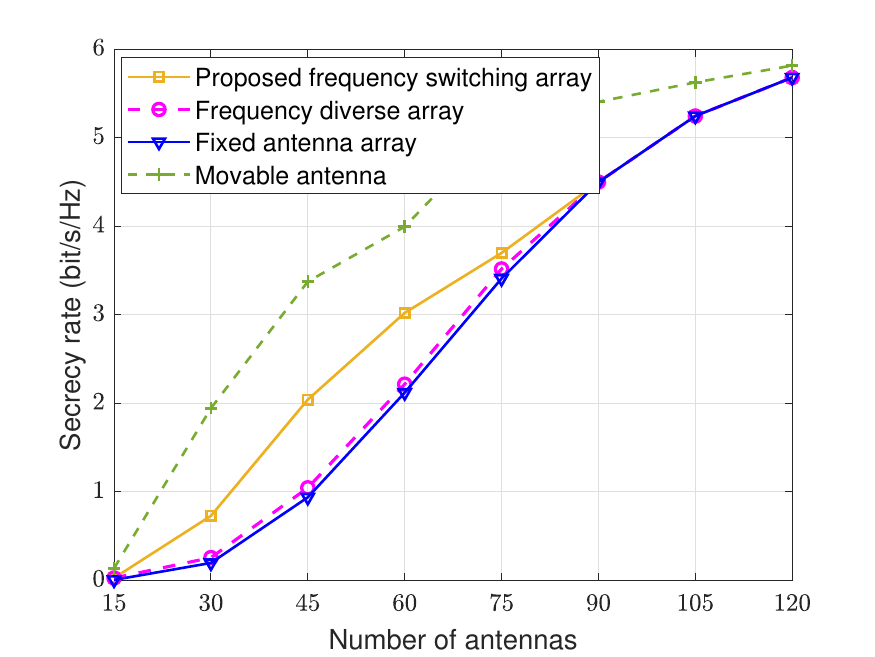}
	\caption{Secrecy rate versus number of antennas} 
	\label{Fig:antenna}
	\vspace{-10pt}
\end{figure}
\begin{figure}[t]
	\centering
	\vspace{-14pt}
	\includegraphics[width=0.38\textwidth]{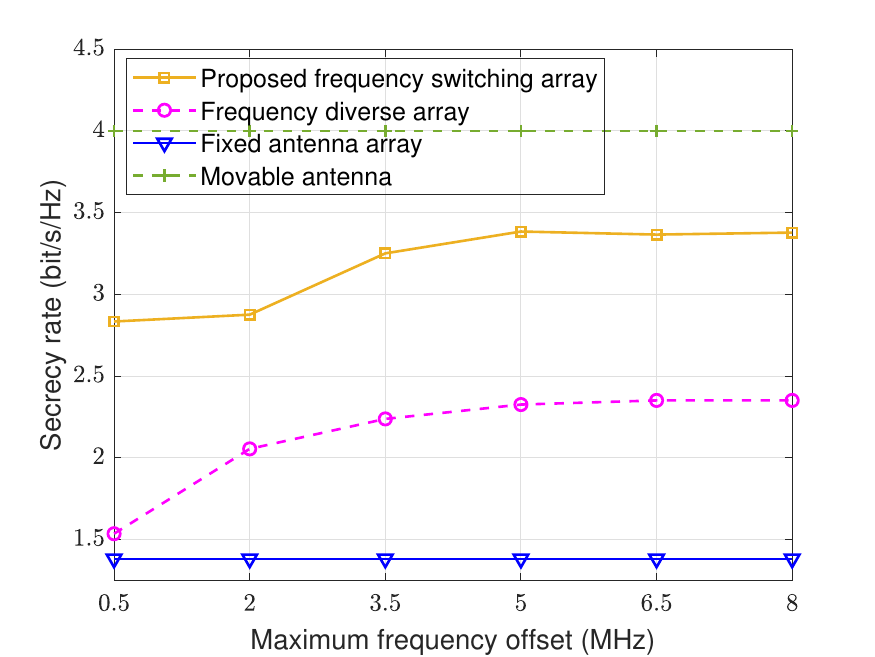}
	\caption{Secrecy rate versus frequency offset} 
	\label{Fig:offset}
	\vspace{-5pt}
\end{figure}
\begin{figure}[t]
	\centering
	\vspace{-10pt}
	\includegraphics[width=0.38\textwidth]{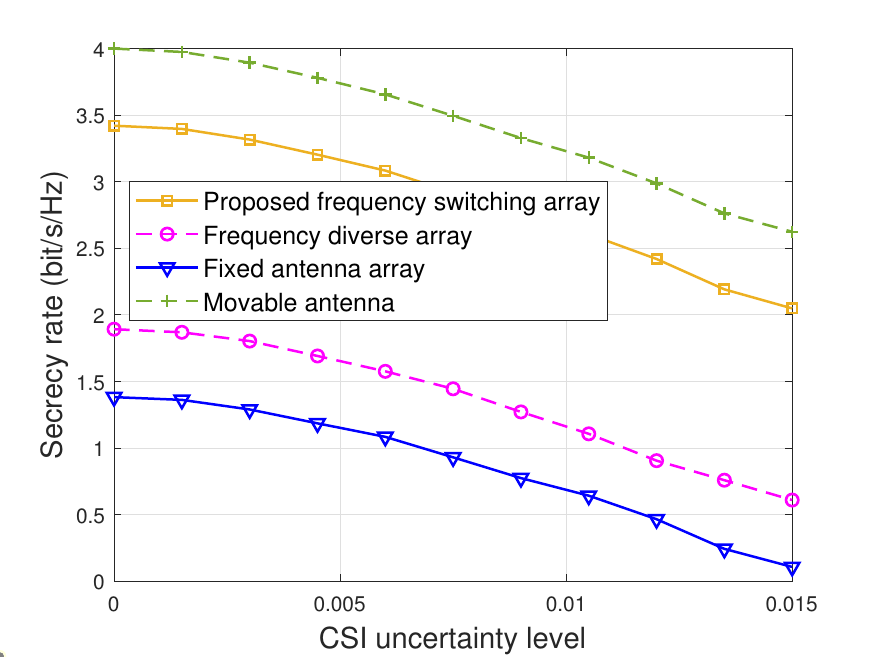}
	\caption{Secrecy rate versus CSI uncertainty level} 
	\label{Fig:bound}
	\vspace{-10pt}
\end{figure}

\vspace{-10pt}
\subsection{Performance Comparison}
In Fig. \ref{Fig:convergence}, we plot the convergence behavior of proposed optimization algorithm for the cases of $N \in \{15, 21\} $ and $P_{\max} \in \{24, 30\}$ dBm.
It is observed that the proposed algorithm converges to a stationary point within $100$ iterations for all cases, showcasing the convergence of optimization algorithm.
In addition, a larger number of antennas (or transmit power) requires more iterations. This is because a greater number of antennas results in a higher achievable secrecy rate, which, under the given convergence threshold and gradient descent step size, requires more iterations to reach convergence.

In Fig. \ref{Fig:Power}, we plot the achievable secrecy rate versus transmit power with $N = 15$. Several important observations are made as follows.
First, the achievable secrecy rates of all schemes increase with the transmit power.
Second, among all benchmark schemes, the proposed FSA achieves a higher secrecy rate than both FPAs and FDAs, and is only slightly inferior to that of MAs.
This is because, although the proposed FSA has comparable beam-control capability to MAs, it suffers from power attenuation due to the increased carrier frequency.
However, as compared to FDAs and FPAs, the proposed FSA offers higher beam-control flexibility, leading to better secrecy-rate performance.
Finally, it is observed that at the low-SNR (i.e., transmit power) regime, the secrecy-rate performance gain of the proposed FSA and MAs over FDAs and FPAs is not significant.
This is because, under low-SNR conditions, the noise power dominates, and the gain from controlling null-steering over Eves is negligible as compared to noise power.

In Fig. \ref{Fig:antenna}, we plot the achievable secrecy rate versus the number of antennas.
First, it is observed that the achievable secrecy rates for all schemes increase with the number of antennas at Alice. 
Second, when the number of antennas is sufficiently large, the achievable secrecy rates of all schemes are nearly identical.
This can be intuitively understood, as the channel correlation between the two Eves and Bob becomes weaker even under the FPA architecture, making it difficult for the Eves to effectively intercept confidential information.
As a result, adjusting frequency or antenna positions no longer provides significant secrecy-rate performance gain.

In Fig. \ref{Fig:offset}, we plot the achievable secrecy rate versus maximum frequency offset.
Several key observations are made as follows.
First, the achievable secrecy rates of the proposed FSA and the FDA increase with the maximum frequency offset, while those of MAs and FPAs remain unchanged.
This is intuitively expected since a larger frequency offset allows the FSA and the FDA to more flexibly control the beam pattern, enabling null-steering over the two Eves.
In contrast, the secrecy-rate performance of MAs and FPAs is independent of the maximum frequency offset.
%Moreover, when the maximum frequency offset is sufficiently large, the proposed FSA and FDA share similar secrecy-rate performance, because enough frequency offsets can already effectively suppress eavesdropping, making further changes to the carrier frequency unnecessary.

In Fig. \ref{Fig:bound}, we plot the secrecy rate versus CSI uncertainty level.
We assume that the channel estimation error defined in \eqref{eq:estimation error} satisfies $\Delta\mathbf{h}_{{\rm E},m} \sim \mathcal{C N}\left(0, \mathbf{\Sigma}_m\right)$.
Herein, $\boldsymbol{\Sigma}_m=\varepsilon_m^2|\hat{\mathbf{h}}_{{\rm E},m}|_2^2 \mathbf{I}_N$, where $\varepsilon_m \in [0,1]$ represents the so-called CSI uncertainty level.
From Fig. \ref{Fig:bound}, it is shown that with the increase of CSI uncertainty level, the secrecy rates of all schemes decrease, which highlights the necessity of robust beamforming design.
Moreover, the proposed FSA consistently outperforms both the fixed array and FDA in terms of secrecy rate, thereby demonstrating its effectiveness in enhancing PLS and its superior robustness.

\vspace{-5pt}
\section{Conclusions}
In this paper, we proposed a new FSA to enhance the performance of PLS systems, which can effectively reduce the hardware and signal processing cost of MAs, while enjoying flexible DoFs in the beam control.
An optimization problem was formulated to maximize the secrecy rate at Bob by jointly designing the transmit beamformer, carrier frequency, and the FIV.
For ease of analysis, we first considered a special case where one Eve and MRT-based beamformer were adopted and then a secrecy-guaranteed problem with null-steering towards Eve was formulated.
We obtained a closed-form optimal solution to the secrecy-guaranteed problem and revealed that thanks to the range-dependent beamforming, the proposed FSA could flexibly realize null-steering over Eve even when Bob and Eve were located at the same angle.
Then, to solve the non-convex problem for the general case, we designed an efficient optimization algorithm by leveraging the BCD technique and PGA method.
Numerical results demonstrated the convergence of the proposed optimization algorithm and confirmed that the proposed FSA can effectively enhance PLS performance as compared to various benchmark schemes.

\bibliographystyle{IEEEtran}
\bibliography{IEEEabrv}

\end{document}